\documentclass[twoside,leqno,twocolumn]{article}

\usepackage[utf8]{inputenc}
\usepackage[T1]{fontenc}
\usepackage{booktabs}
\usepackage{etex}
\usepackage{multirow}
\usepackage{amssymb}
\usepackage{amsmath,amsfonts,amsthm}
\usepackage{pdfpages}
\usepackage{graphicx}
\usepackage{url}
\usepackage{hyperref}
\usepackage{enumerate}
\usepackage{tikz}
\usetikzlibrary{decorations.markings, calc}
\usepackage{stmaryrd}
\usepackage{color}

\usepackage{wrapfig}
\usepackage{graphicx}
\usepackage{enumerate}
\usepackage{tikz}
\usepackage{stmaryrd}
\usepackage{color}
\usepackage{pdfpages}

\usepackage[]{algorithm}
\usepackage{etoolbox}\AtBeginEnvironment{algorithmic}{\footnotesize} 
\usepackage[noend]{algpseudocode}
\usepackage{mathrsfs}
\usepackage{caption}

\newcommand{\B}{\mathcal{B}}

\newcommand{\G}{\mathcal{G}}

\newcommand{\M}{\mathcal{M}}
\newcommand{\N}{\mathcal{N}}
\newcommand{\T}{\mathcal{T}}
\newcommand{\U}{\mathcal{U}}

\renewcommand{\S}{\mathcal{S}}

\newcommand{\X}{\mathcal{X}}

\newcommand{\Z}{\mathcal{Z}}

\newcommand{\IE}{\mathbb{E}}
\newcommand{\IV}{\mathbb{V}}

\newcommand{\IP}{\mathbb{P}}

\DeclareRobustCommand{\stirling}{\genfrac\{\}{0pt}{}}

\DeclareMathOperator*{\Seq}{\mbox{\sc Seq}}
\DeclareMathOperator*{\SeqUn}{\mbox{\sc Seq}_{\geq 1}}
\DeclareMathOperator*{\SeqDeux}{\mbox{\sc Seq}_{\geq 2}}
\DeclareMathOperator*{\Set}{\mbox{\sc Set}}

\newtheorem{theorem}{Theorem}
\newtheorem{proposition}[theorem]{Proposition}

\newtheorem{lemma}[theorem]{Lemma}
\newtheorem{corollary}[theorem]{Corollary}

\theoremstyle{definition}
\newtheorem{Definition}[theorem]{Definition}

\title{Ranked Schröder Trees\thanks{
		This research is partially supported by the ANR project {\em MetACOnc}, ANR-15-CE40-0014.}
}

\author{Olivier Bodini\thanks{
    Laboratoire d'Informatique de Paris-Nord,
	CNRS UMR 7030 - Institut Galil\'ee - Universit\'e Paris-Nord,
	99, avenue Jean-Baptiste Cl\'ement, 93430 Villetaneuse, France.
	Email: \url{{Olivier.Bodini, Mehdi.Naima}@lipn.univ-paris13.fr}.}\\
	\and
	Antoine Genitrini\thanks{
	    Sorbonne Universit\'e, CNRS,	
	    Laboratoire d'Informatique de Paris 6 - LIP6 - UMR 7606, F-75005 Paris, France.
		Email: \url{Antoine.Genitrini@lip6.fr}.}\\
	\and
	Mehdi Naima$^\dag$}

	\date{}

\begin{document}

\maketitle

%\pagenumbering{arabic}
%\setcounter{page}{1}%Leave this line commented out.

\begin{abstract}\small\baselineskip=9pt
    In biology, a phylogenetic tree is a tool to represent the evolutionary relationship between species. 
    Unfortunately, the classical Schröder tree model is not adapted
    to take into account the chronology between the branching nodes.
    In particular, it does not answer the question: how many different phylogenetic 
    stories lead to the creation of $n$ species and what is the average time to get there? 
    In this paper, we enrich this model in two distinct ways in order to
    obtain two ranked tree models for phylogenetics, i.e. models coding chronology.

    For that purpose, we first develop a model of (strongly) increasing
    Schröder trees, symbolically described in the classical context of
    increasing labeling trees. Then we introduce a generalization for the
    labeling with some unusual order constraint in Analytic Combinatorics
    (namely the weakly increasing trees).

    Although these models are direct extensions of the Schröder tree model,
    it appears that they are also in one-to-one correspondence with several
    classical combinatorial objects.
    Through the paper, we present these links, exhibit some parameters in
    typical large trees and conclude the studies with efficient uniform
    samplers.\\
    
\end{abstract}

\section{Introduction}
\label{sec:intro}

In biology a \emph{phylogenetic tree} is a classical tool to represent
the evolutionary relationship among species.
At each bifurcation, or multifurcation, of the tree, the descendant
species from distinct branches have distinguished themselves in some manner.

One of the first illustrations of an evolutionary tree was made by Darwin
in his book \emph{On The Origin of Species}~\cite{darwin1859}.
His idea was to represent the divergence of characters and species. 
Multifurcations represent a well-marked variety of a certain kind and
this process then continues on the new varieties and so on.
Interest grew in tree evolutions as these models give insight on how species  evolved.
Different tree models were proposed with the idea of finding trees
that fits best nowadays observations and data sets. 
These models of graphs include rooted, unrooted, labeled, unlabeled,
bifurcating or multifurcating trees or networks. 
By defining some metrics between these models, people develop algorithms
focusing on state space exploration or on tree inference.
For details on tree models in phylogenetics and inference algorithms
see the book of Felsenstein~\cite{Felsenstein2003}
and the one of Steel~\cite{Steel2016PhylogenyD} for a more recent survey
with combinatorial aspects also.
Thanks to the development of bioinformatics many tools have thus emerged,
in order to build automatically such tree diagrams.
Some examples of programs are \emph{PHYLIP}, a tool  for inferring 
phylogenetic trees~\cite{felsenstein1993phylip} or \emph{PAML} 
that is phylogenetic analyser based on the maximum likelihood~\cite{doi:10.1093/molbev/msm088}.
In order to develop these new tools several structural studies
have been realized to model correctly the fundamental parameters 
defined by biologists.

In 1870 Schröder presented an original model published 
into the paper \emph{Vier combinatorische Probleme}~\cite{Schroder70}.
The fourth problem presents a phylogenetic tree model enumerating
trees by their number of leaves.
See for example~\cite{felsenevo} for the phylogenetic interpretation.

While it has been highlighted that this first model is not adapted
to take into account the chronology between branching nodes belonging to two
distinct fringe subtrees, other approaches have been developed
to consider such a history of the evolution process.
In particular in the context of binary trees, we can mention
the stochastic model of Yule~\cite{Yule25} and its generalization by 
Aldous~\cite{Aldous96}.
Such tree models, including history evolution, are usually called \emph{ranked tree
	models} in phylogenetics.
But these new models are not based on the original Schröder tree model.
%But the classical Schröder trees does not take into account
%the time evolution or the chronology of the 
%considered system.
To the best of our knowledge, there seems to have been no attempt to
enrich Schröder's original model so as to encode the chronology of evolution.
\begin{figure}[ht]
	\begin{center}
		\begin{tabular}{c c}
			\begin{minipage}{2cm}
				\includegraphics[scale=0.2]{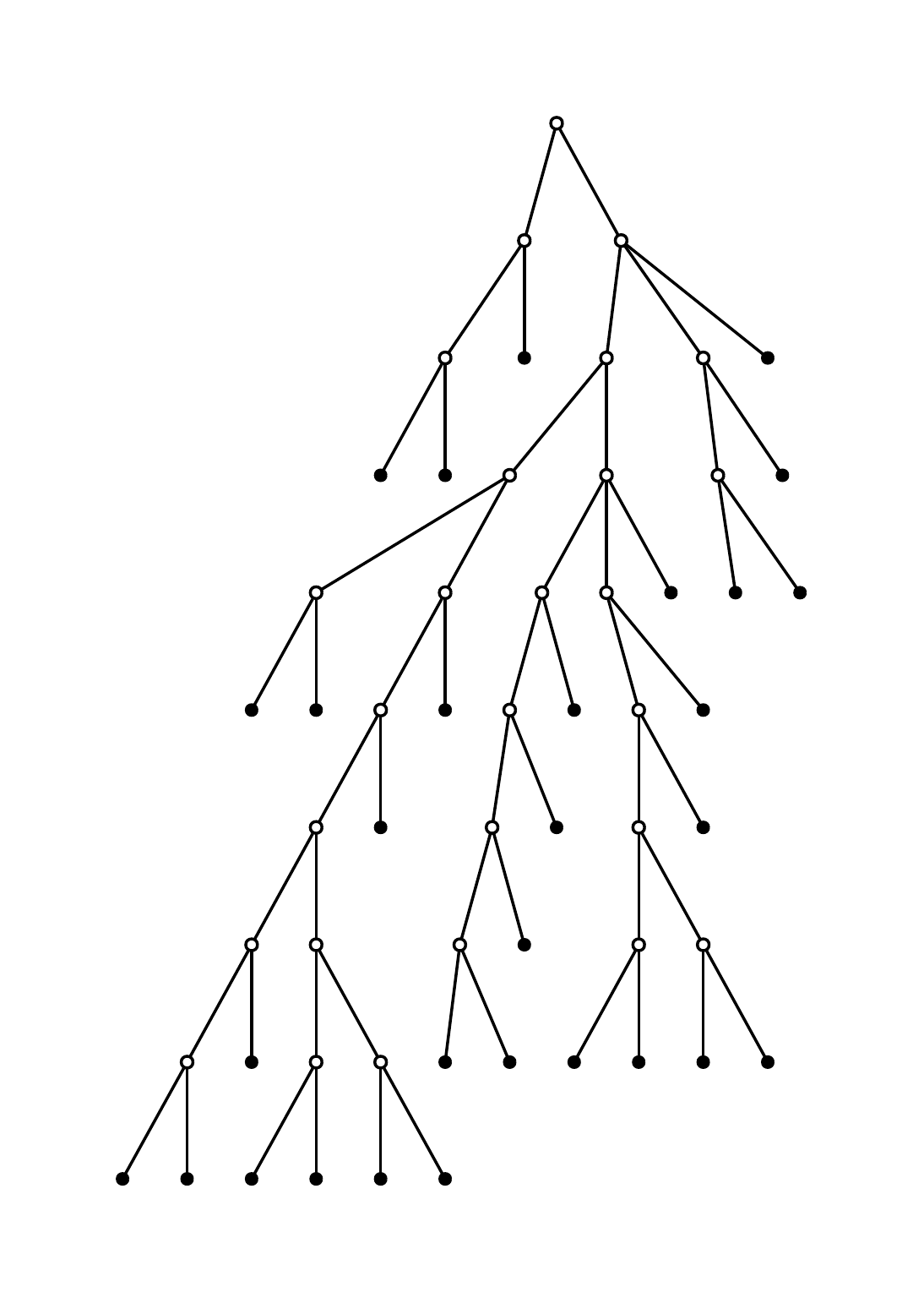}
				
				\vspace*{5.6cm}
				$\;$
			\end{minipage}
			&
			\begin{minipage}{5cm}
				\includegraphics[height=9cm, width=5cm]{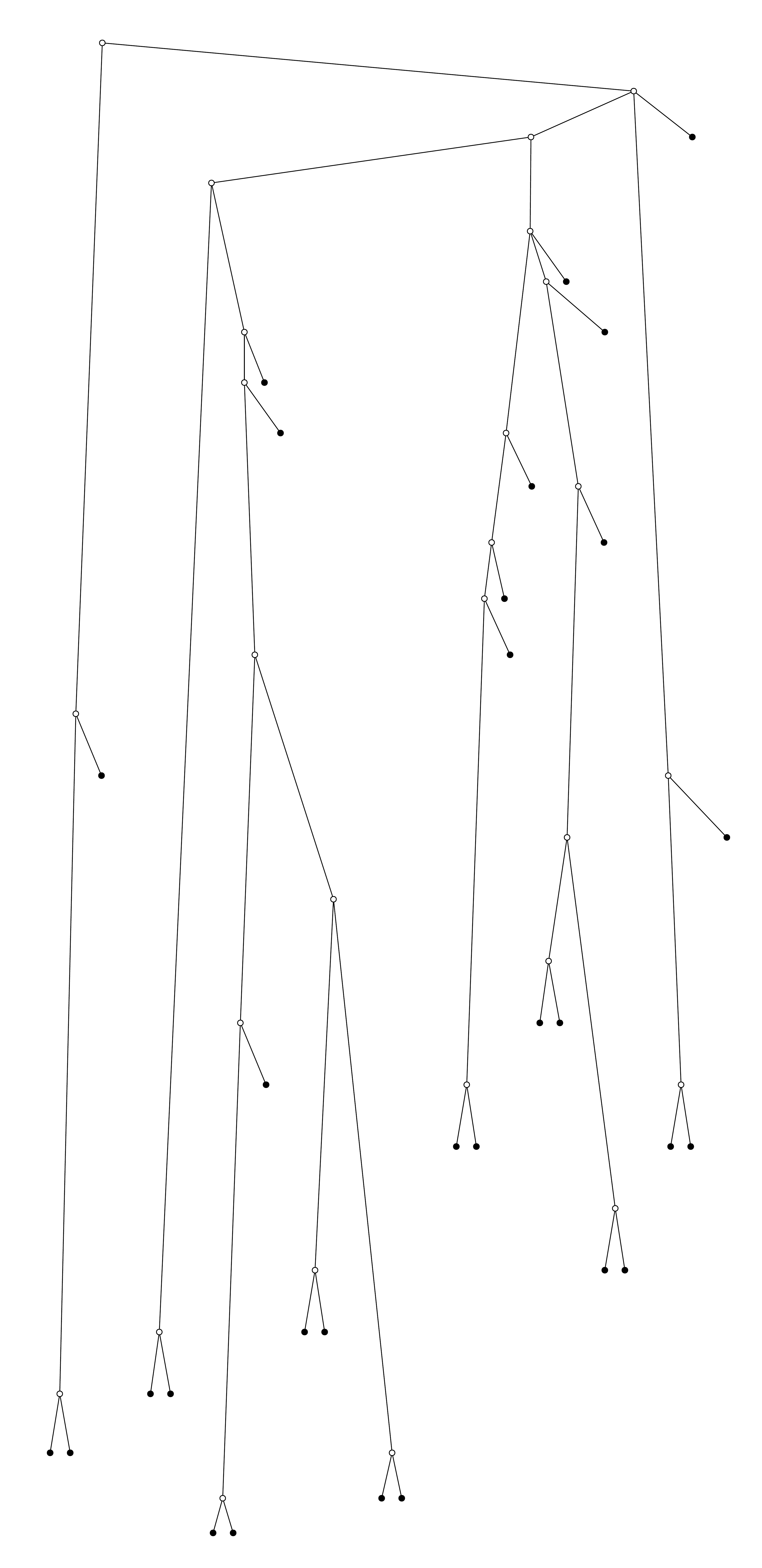}
			\end{minipage}
		\end{tabular}
		\caption{\label{fig:phylo}
			A Schröder tree: without chronological evolution (on the left handside),
			and with chronological evolution (on the right handside)}
	\end{center}
\end{figure}

So, the main goal of this paper consists in designing ranked tree models
based on the classical Schröder structure.
In Figure~\ref{fig:phylo} we have represented the same phylogenetic tree
on the left handside as a classical Schröder tree, and on the right handside
as a \emph{strongly increasing Schröder tree}, the first model we develop in this paper.

A first natural idea in this direction consists in considering the model of a \emph{recursive tree}.
Such a structure is a rooted labeled tree, whose root is labeled by $1$
and the successors of a given node, with label $\nu$, have a label greater than $\nu$.
Each integer between 1 and the total number of nodes is present once in the tree. 
Many variations of this model have been presented in the literature: 
see~\cite{Drmota09} and the references therein.
In this context, we are able to define a simple \emph{evolution process}
that allows to build very efficiently large trees with simple iterative rules.
Furthermore, usually the history of construction is naturally kept in the final
large tree through the increasing labeling.
It is also important to note that apparently minor changes on the growth rules
induce drastic differences in the typical properties of the considered models.
See for example the book
of Drmota~\cite{Drmota09} that presents many extensions of the classical
model (e.g. plane oriented recursive trees, fixed arity -- or out-degree -- recursive trees)
and details several quantitative studies for different fundamental parameters
like the profile of such tree models.

Let us recall the sample of a recursive tree
(uniformly for all trees of the same size, i.e. the same number of nodes):
\emph{start with the single size-$1$ tree}, reduced to a root,
and iterate: \emph{at step $n\in\{2,3,\dots\}$ choose uniformly a node in the tree
	under construction (labeled with an integer between $1$ and $n-1$)
	and attach to it a new node labeled by $n$}.

While many variations on these models have been studied, it is very interesting
to note that the increasing version of Schröder trees
seems not to have been analyzed. Our model is also very natural due to its
similarities to the probabilistic model of Yule trees (cf. e.g. \cite{MS01}) that 
take into account the chronological mutations of species.

%Recall that the classical phylogenetic tree model is
%related to Schröder's second problem that consists
%in the enumeration of the bracketed expressions
%of a string containing $n$ times the letter $x$. Formally, the string $(x)$
%is well bracketed and if $\sigma_1, \sigma_2 \dots, \sigma_r$ are well bracketed, then 
%the expression $(\sigma_1 \sigma_2 \dots \sigma_r)$ is also well bracketed.
%Obviously such expressions can be isomorphically related to trees
%whose leaves contains the letter $x$ and the internal nodes encode the bracketing.
We develop in this paper two distinct models for phylogenetic trees
satisfying in priority two new constraints: (1) to take into account the chronological
evolution and (2) to be efficient to simulate. Both models are based on some
\emph{increasingly labeling} of Schröder tree structures.

In this paper, we are focusing on the distinct histories possible
for a fixed number of final species. From a graph model point of view,
it consists in the quantitative study of the number of structures of a given size.
Furthermore, beyond some characteristics shared by our model and recursive trees,
or increasing fixed arity trees, we will point out several relations
to other classical combinatorial objects, in particular permutations, Stirling numbers.
Due to the many links to combinatorial objects, increasing 
Schröder trees are thus interesting in themselves as combinatorial structures.

The paper is organized as follows.
In Section~\ref{sec:high_inc} we introduce formally our first ranked phylogenetic tree model
and introduce a non classical point of view for the tree specification.
We present the enumeration of the trees and relate them to permutations.
Then we compute important parameters
of the model. We conclude this section with
the presentation of a linear algorithm
for the uniform sampling of trees.
Section~\ref{sec:weak_inc} is devoted to our second model for ranked
phylogenetic trees.
It is based on a non-classical way of increasingly labeling a tree structure.
The section is composed like the first one: after the enumeration
of the trees, we relate them to classical combinatorial objects, derive
some tree parameters and we finally conclude the section
with an efficient unranking algorithm for the uniform sampling of our trees.

Some technical proofs are detailed in the appendix due to obtain a clear paper structure.

\section{Strongly Increasing Schröder trees}
\label{sec:high_inc}

The first model we develop is based on a
almost classical notion of increasing labeling
in Analytic Combinatorics.

\subsection{The model and its context}
\label{sec:model}

The tree structure associated to strongly increasing Schröder tree
corresponds to \emph{Schröder trees}, i.e. the combinatorial class of 
\emph{rooted plane\footnote{A plane tree is such that the children of a node are ordered.}
	trees whose internal nodes have arity at least $2$}.
The reader can refer to~\cite[p. 69]{FS09} for some details.
The size of a Schröder tree is the number of leaves in the tree.
Note that in the tree structure neither the internal nodes, nor the leaves are labeled.
The combinatorial class~$\S$ of Schröder trees is specified as
$\S = \Z \cup \SeqDeux \S,$
that translates, via the classical \emph{symbolic method}
presented by Flajolet and Sedgewick~\cite{FS09},
into the following equation, 
$S(z) = z + \frac{S(z)^2}{1-S(z)}$,
satisfied by its ordinary generating function
$S(z) = \sum_{n\geq 1} s_n z^n$
where $s_n$ is the number of structures of size $n$ (i.e. with $n$ leaves).

In this section, we are interested in an increasingly labeled variation of Schröder trees.
\begin{Definition}
	A \emph{strongly increasing Schröder tree} has a tree structure that is a Schröder tree
	and moreover its internal nodes are labeled with the integers between $1$ and $\ell$ 
	(where $\ell$ is the number of internal nodes),
	in such a way that all labels are distinct and the sequence of labels
	in each path from the root to a leaf is increasing.
\end{Definition}
Note, in the Analytic Combinatorics context, 
such a labeling of trees is called
\emph{increasing labeling} (without the term \emph{strongly}).
In order to distinguish
clearly this first model from the second one presented in Section~\ref{sec:weak_inc}
we have added this term. 
But from here, inside this section we will use the classical
denotation increasing tree.
\begin{figure}[ht]
	\begin{center}
		\includegraphics[scale=0.4]{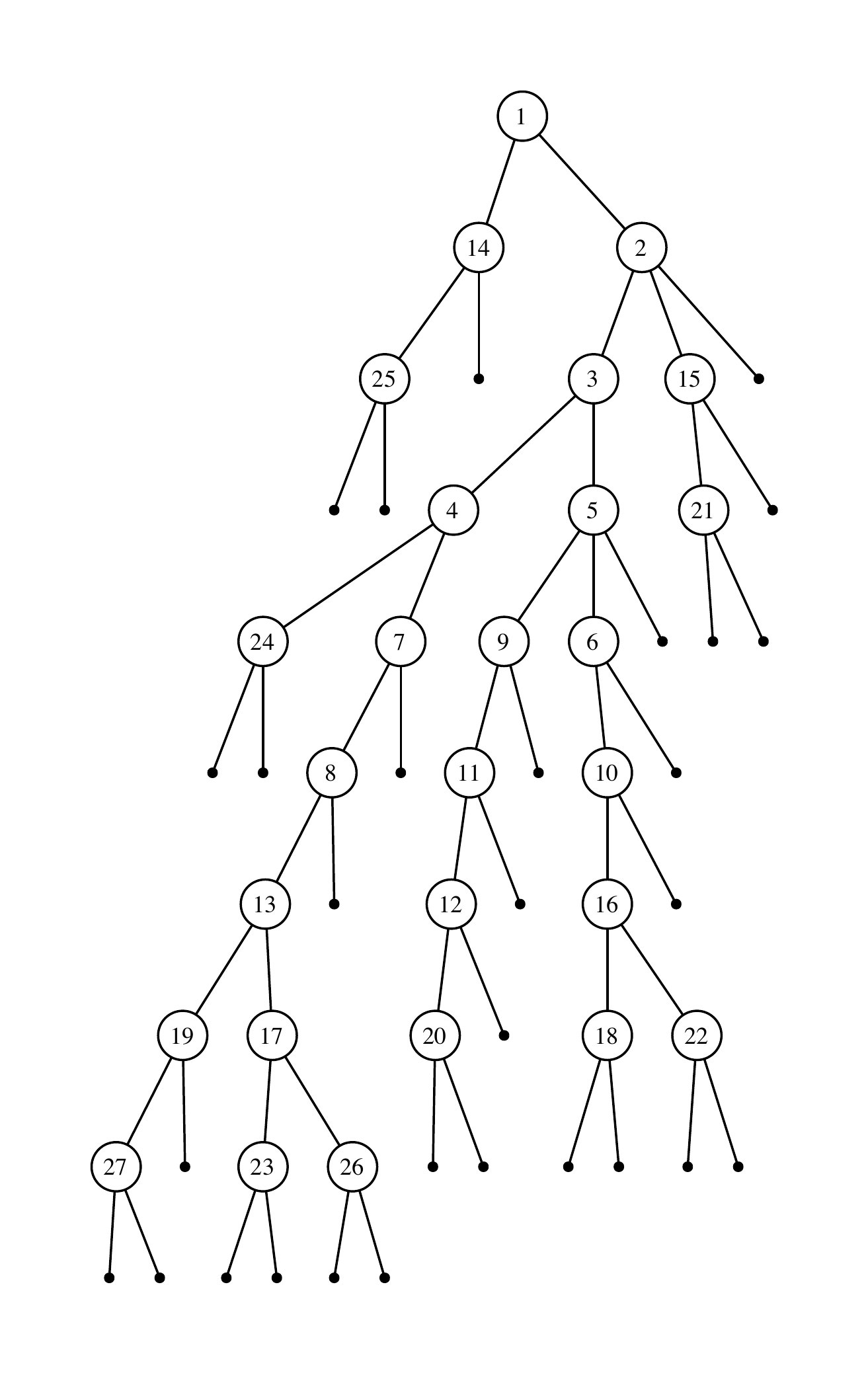}
		\caption{\label{fig:inc_tree}
			A strongly increasing Schröder tree}
	\end{center}
\end{figure}
Trees that are increasingly labeled can be in a certain extent
specified with the Greene's operator $\;^{\square} \star$
% by
%$\S^+ = \Z ^{\square} \star \SeqDeux(\S^+) $ 
(cf. for example~\cite[p. 139]{FS09}).
Then the specification is translated into an equation
satisfied by the exponential generating function.
But in our context, the size of a tree is the number of leaves
(which corresponds to the final number of species),
and the increasingly labeling constraint
is related to the internal nodes.
We specify this class by using a second variable $u$ to mark the internal nodes.
\[
S^*(z,u) = \sum_{n,\ell} s_{n,\ell} z^n \frac{u^\ell}{\ell!} 
= z + \int_{v=0}^{u} \frac{S^*(z,v)^2}{1-S^*(z,v)}\mathrm{d}v.
\]
While the integral equation could be analyzed further, 
we prefer, in the following, to introduce an alternative way
to define our objects. This new approach is easier to handle and
it also naturally extends to define our second model of trees
developed in Section~\ref{sec:weak_inc}, namely the weakly increasing Schröder trees.

In Figure~\ref{fig:inc_tree} we have represented an increasing Schröder tree of size $30$
with $27$ internal nodes. 
%(uniformly sampled among all increasing trees of size $30$).
This increasing tree is the same tree as the one represented in Figure~\ref{fig:phylo}
with the chronological evolution,
where the internal node labeled by $\ell$ is laid on level $\ell-1$, 
for all $\ell\in\{1, \dots, 27\}$.

In order to describe the building of increasing Schröder trees,
we introduce an \emph{evolution process}. It consists in an iterative
way that substitutes a leaf by an internal node attached to several leaves. More formally:
\begin{itemize}
	\item \emph{Start with a single (unlabeled) leaf;}
	\item Iterate the following process: \emph{at step $\ell$ (for $\ell\geq 1$),
		select a leaf and replace it by an internal node with label $\ell$
		attached to a sequence of at least two leaves.
	} 
\end{itemize}
Remark that the increasing labeling corresponds to the chronology 
of the tree building.

\subsection{Exact enumeration and relationship with permutations}
\label{sec:enum}

Let us denote by $\T$ the class of increasing Schröder trees. 
By using the evolution process 
we exhibit a specification for $\T$ as follows:
\begin{equation}\label{eq:specT2}
\T = \Z \cup \left(\Theta \T \times \SeqUn (\Z) \right).
\end{equation}
In this specification, $\Z$ stands for the leaves, and 
the operator $\Theta$ is the classical pointing operator 
(cf. in~\cite[p. 86]{FS09} for details).
The specification is a direct rewriting of the evolution process.
A tree is either reduced to a leaf or
at each step an atom (i.e. a leaf) is pointed
in the tree under construction and is replaced by an internal node
(whose labeling is deterministic: it corresponds
to the step number) attached to a sequence of at least two leaves
(the one that has been pointed is reused as the leftmost child,
it is the reason why the operator $\SeqUn$ does not 
contain the empty sequence and starts with sequences containing one element).

The symbolic method
translates this specification into a functional equation satisfied
by the generating series associated to the combinatorial class. 
Note that the functional equation is satisfied by
the \emph{ordinary generating series} associated to $\T$: $T(z) = \sum_{n\geq 1} t_n\; z^n$.
The increasing labeling is here transparent and thus the objects seems not labeled
(in fact, the leaves, marked by $\Z$ are really unlabeled):
\begin{equation} \label{eq:funeq}
T(z) = z+ \frac{z^2}{1-z} \; T'(z).
\end{equation}
By extracting the coefficients of the series, we derive the two following recurrences.
\begin{equation} \label{eq:rec}
\left\{ \begin{array}{l}
t_1 = 1, \quad t_2 = 1,\\
\text{and for } n>2,\\
t_n = n \cdot t_{n-1}.
\end{array}
\right.
%\hspace*{3cm}
\qquad
\left\{ \begin{array}{l}
t_1 = 1,\\
\text{and for } n>1,\\
t_n = \sum_{k=1}^{n-1} k\cdot t_k .
\end{array}
\right.
\end{equation}
Both recurrences are computed thanks to equation~\eqref{eq:funeq}.
The direct extraction $[z^n]\; T(z)$ exhibits the rightmost recurrence.
This recurrence exhibits that the calculation of the $n$-th term is of quadratic complexity
(in the number of arithmetic operations).
The leftmost recurrence is obtained by 
extracting $[z^n]\; (1-z)\cdot T(z)$ and then by simplifying the resulted equation.
Here the calculation of the $n$-th term is of linear complexity.

%\begin{table}[h]
%  \centering
%
%\begin{tabular}{ |l |l|c|r|l | l | l | l| l | l | l | l |}
%\hline
%$n$ & $0$ & $1$ & $2$ & $3$  & $4$ & $5$ & $6$ & $7$ & $8$ & $9$ & $10$\\
%
%\hline
%
%$t^{[2]}_n$ & $0$ & $1$ & $1$ & $3$  & $12$ & $60$ & $360$ & $2520$ & $181440$ & %$1814400$ & $19958400$\\
%\hline
%
%
%\end{tabular}
%\caption{First terms of $T_n^{[2]} $ }
%\end{table}

Thus we directly prove $t_n = n! / 2$ 
for all $n\geq 2$. The sequence $(t_n)_n$ appears  under the reference
\href{https://oeis.org/A001710}{\texttt{OEIS A001710}}\footnote{Throughout this paper,
	a reference \texttt{OEIS A$\cdots$} points to Sloane’s Online Encyclopedia of Integer
	Sequences \url{www.oeis.org}. }.
Observing the growth rate of $(t_n)_n$ proves that the 
ordinary generating series $T(z)$ is formal: its radius of convergence is~$0$.

\subsection{Analysis of typical parameters}
\label{strong:iteration}

Here we are interested in the quantitative
study of four distinct parameters of increasing Schröder trees.
The first one corresponds to the number of internal (labeled) nodes
of a size-$n$ tree. This fundamental parameter corresponds
to the number of steps in the evolution process that are
necessary to build the given tree. Recall the arity of internal
nodes is at least two, thus this parameter is not deterministic.
The second and the third parameters are related to the root node.
We study its arity and the number of leaves attached to it in a typical
tree of size $n$. 
But in a tree of size $n$ (tending to infinity)
all internal nodes whose labels are independent from~$n$ have the 
same characteristics than the root: thus these two parameters are also
important for the global quantitative aspects of a large tree.
Finally, the fourth parameter corresponds to the typical number
of binary nodes in a large tree. This study becomes natural
once we have seen the typical value of the number of internal nodes
of a large tree.

\subsubsection*{Quantitative analysis of the number of iteration steps}

A fundamental parameter characterizing the increasing Schröder trees
is their number of internal nodes.
This parameter is interesting in itself, but furthermore
it corresponds to the maximal label value in the tree, and thus
it is also the number of steps of the building process.

To study both the number of internal nodes and the number of leaves,
we enrich the specification~\eqref{eq:specT2} with an additional parameter~$\U$
marking the internal nodes. 
%We also translate it into the generating series's context. 
\begin{align}\label{eq:specbiv_T2}
&\T = \Z \cup \left(\U \times \Theta_\Z \T \times \SeqUn(\Z) \right); \nonumber \\
%    \hspace*{2cm}
&T(z,u) = z + \frac{u\; z^2}{1-z} \; \partial_z T(z,u).
\end{align}
The operator $\Theta_\Z$ consists in pointing an element marked by $\Z$.
The partial differentiation according to $z$
is written as $\partial_z \cdot$.
With the notation
$T(z,u) = \sum_{n\geq 1} t_n(u) z^n$,
the equation~\eqref{eq:specbiv_T2} gives two recurrences
satisfied either by $(t_n(u))$, 
or by $(t_{n,k})$, where $t_{n,k}$ is the number of trees with $n$
leaves and $k$ internal nodes (that are increasingly labeled):
%{\small
\begin{align} \label{eq:rec2}
&\left\{ \begin{array}{l}
t_1(u) = 1, \qquad t_2(u) = u,  \\
\text{and if } n> 2,\\
t_{n}(u) = (1+(n-1)u) t_{n-1}(u);
\end{array}
\right. \nonumber \\
%    \hspace*{0.6cm}
&\left\{ \begin{array}{l l}
t_{n,k} = t_{n-1,k} + (n-1) \; t_{n-1,k-1} & \text{ if } 0 < k < n\\
t_{1,0} = 1, \qquad t_{n,1} = 1 &\text{ if } n > 1 \text{ and} \\
t_{i,j} = 0 & \text{ otherwise}.
\end{array}
\right. 
\end{align}
%}
\begin{wrapfigure}[14]{r}{5.2cm}
	\begin{center}
		%\begin{table}[htb]
		\begin{tabular}{ l l l l l l l l l l l }
			$1$ ,& & & & & & & & & & \\
			$0$ ,&  $1$, &  & & & & & & &  &\\
			$0$ ,&  $1$, & $2$, &   & & & & & &  &\\
			$0$ ,&  $1$, & $5$, & $6$,  &   & & & & &  &\\
			$0$ ,&  $1$, & $9$, & $26$,  &  $24$, &    & & & &  &\\
			$0$ ,&  $1$, & $14$, & $71$,  &  $154$, &  $120$,  &   & & &  &\\
			%   $0$ ,&  $1$, & $20$, & $155$,  &  $580$, &  $1044$,  & $720$,  &  & & & \\
			%   $0$ ,&  $1$, & $27$, & $ 295$,  &  $1665$, &  $5104$,  & $8028$,  & $5040$,  &  & & \\
			%   $0$ ,&  $1$, & $35$, & $ 511$,  &  $4025$, &  $18424$,  & $48860$,  & $69264$,  & $40320$, &  & \\
			%   $0$ ,&  $1$, & $44$, & $ 826$,  &  $8624$, &  $54649$,  & $214676$,  & $509004$,  & $663696$, & $362880$, &  \\
		\end{tabular}
		\caption{\label{tab:interne_leaf1}
			Distribution of $t_{n,k}$ for size-$n$ trees, $n \in \{1, 2, \dots, 6\}$, 
			of the number of internal nodes $k \in \{0,1,\dots, n-1\}$}
	\end{center}
\end{wrapfigure}

Remark that the extremal conditions are trivially obtained through our
construction in particular the sequence $(t_{n,n-1})_n$ is enumerating
increasing binary trees.
Once again, these efficient recurrences are obtained thanks to the
extraction of $[z^n] (1-z)\cdot T(z,u)$.
In Figure~\ref{tab:interne_leaf1}, for the tree size-$n$ from $1$ to $6$,
we present the distribution of the number of trees according to their number $k$ of
internal nodes.

The \emph{Borel transform}, denoted as $\B\cdot$, translates
an ordinary generating series into 
its analog exponential generating series. For example, we obtain
$\B T(z) = \sum_{n\geq 1} t_n \; \frac{z^n}{n!}$.
In particular, due to the growth of the coefficients $(t_n)_n$ we directly
observe that $\B T(z)$ is analytic around $0$ (with radius of convergence $1$).

\begin{proposition}\label{prop:bivt} 
	The Borel transform on $T(z,u)$ relatively to the variable $z$ gives
	\[
	\B T(z,u)= \sum_{n\geq 1} \sum_{k=0}^{n-1} t_{n,k} \; u^k \; \frac{z^n}{n!}
	=\frac{u(1-zu)^{-\frac{1}{u}} -u +z}{1+u}.
	\]
\end{proposition}
Here we just present the key-ideas of the proof, but
details are given in Appendix~\ref{app:high_inc}.
\begin{proof}[Key-ideas]
	Applying the Borel transform on equation~\eqref{eq:specbiv_T2}
	and then classical properties of the Borel transform for the function $z\cdot f(z)$
	and for the derivative $f'(z)$ yields the result.
\end{proof}

Let us come back to the polynomial $t_n(u)=\sum_{k=0}^{n-1} t_{n,k} \; u^k$.
It corresponds almost to the sequence \href{https://oeis.org/A145324}{\texttt{OEIS A145324}}
related to Stirling numbers.
\begin{corollary}\label{cor:row_extract}
	Let $n\ge 2$. The distribution of the number of internal nodes in increasing
	Schröder trees of size $n$ is
	\[
	t_n(u)=\sum_{k=0}^{n-1} t_{n,k} \; u^k = u \; \prod_{\ell=2}^{n-1} (1+\ell u).
	\]
\end{corollary}
The proof relies on a direct rewriting of the first recurrence
in equation~\eqref{eq:rec2}.
The generating function corresponds to the $n$-th row in the triangle
presented in Figure~\ref{tab:interne_leaf1}.
Although the sequence $(t_n(u))$
is stored in \texttt{OEIS} we exhibit here another link
with a very classical triangle.
By reading each row of the triangle from right to left,
we obtain a shifted version of the triangles 
\href{https://oeis.org/A136124}{\texttt{OEIS A136124,A143491}}.
It corresponds almost 
to the generating function of Stirling Cycle numbers~\cite[p. 735]{FS09}:
$SC_n(u)=\prod\limits_{i=1}^{n-1}(u+i)$. The associated sequence enumerates
size-$n$ permutations that decompose into $k$ cycles, defined as Stirling numbers
of the first kind. % and usually denoted by ${n \brack k}$. 
More formally we prove:
\begin{proposition}
	Defining $\hat{t}_n(u)=\sum_{k=1}^n t_{n,k} \; u^{n-k}$, we obtain
	%	\[ 
	$\displaystyle{\hat{t}_n(u)  = \frac{u}{1+u}SC_n(u)}$.
	%	\]
\end{proposition}
%\begin{proof}
%	\[ \hat{T}_n(u)=u^n \frac{1}{u}\prod\limits_{i=2}^{n-1}(1+i\frac{1}{u})=u^{n-1} \prod\limits_{i=2}^{n-1}(1+i\frac{1}{u})=u\prod\limits_{i=2}^{n-1}(u+i)= \frac{u}{(1+u)}SC_n(u) \]
%\end{proof}
%Keeping this result in mind, we are ready to get by-products for almost free.
%\MP{quels sont les produits gratuits ? exploités ? ou sinon dir qu'on a une nouvelle preuve
%}
Let $\X_n$ be the random variable that maps increasing Schröder trees of size $n$
to their numbers of internal nodes. We want to establish a limit law for 
the distribution $\left(\IP_{\T_n}(\X_n=k)\right)_k$.
But let us first compute its mean and its standard deviation so that
we will then study the convergence of the \emph{normalized}
random variable $\X_n^\star=\frac{\X_n -\IE(\X_n)}{\sqrt{\IV(\X_n)}}$.
We follow here the classical approach presented, for example, in~\cite[p. 157]{FS09}.
Since we consider the uniform distribution among trees of a given size $n$,
we obviously get $\IP_{\T_n}(\X_n=k)=\frac{t_{n,k}}{t_{n}}$.

\begin{proposition}\label{prop:var}
	Let $n\ge 2$, the mean value of $\X_n$ is equal to
	%	\begin{align*}
	%		&\IE_{\T_n}(\X_n) =  n-H_n+\frac{1}{2}, \quad \text{and thus}\\ 
	%		&\IE_{\T_n}(\X_n) = n- \ln n-\gamma+\frac{1}{2}+ O\left(\frac{1}{n}\right),
	%	\end{align*}
	\[
	\IE_{\T_n}(\X_n) =  n-H_n+\frac{1}{2} = n- \ln n-\gamma+\frac{1}{2}+ O\left(\frac{1}{n}\right),
	\]
	with $H_n$ the $n$-th harmonic number and $\gamma$ the Euler constant ($\gamma \approx 0.57721\dots$).
	Furthermore,
	\[
	\IV_{\T_n}[\X_n] = \ln n + \gamma -\frac{\pi^2}{6}-\frac{5}{4} + O\left(\frac{\log n}{n}\right).
	\]
\end{proposition}
Recall that the ordinary generating function for the Harmonic numbers sequence
is  $H(z)=\frac{1}{1-z}\; \ln \frac{1}{1-z}$ (see e.g. \cite[p.~388]{FS09}), then the result
is proved by a direct computation. The proof is presented in Appendix~\ref{app:high_inc}.

This proposition allows us to exhibit the limit law of the distribution $(X^\star_n)$
and proves then that the sequence $(\X_n)$ converges in distribution to a Gaussian law.
%
%Now we can move forward and show the limit law of this distribution, $X^\star_n=\frac{\mathcal{X}_n- \mu_n}{\sigma_n}$ with $\mu_n=n-\log n-\gamma+\frac{1}{2}+o(1)$ and $\sigma_n= \sqrt{ \log n + \gamma -\frac{\pi^2}{6}-\frac{5}{4}}+o(1)$. In the next proposition we will show that $\mathcal{X}_n$ converges to a Gaussian law using an approximation of the Probabilistic Characteristic function \[  \phi_{\mathcal{X}}(t)=\mathbb{E}[e^{it\mathcal{X}}]=\sum\limits_k \mathbb{P}(\mathcal{X}=k)e^{itk}  \].
%
%
\begin{theorem}\label{theo:limit_law_internal}
	Let $\X_n$ be the random variable describing the distribution of
	the number of internal nodes in increasing Schröder trees of size $n$, or equivalently
	the number of building steps to get a size-$n$ tree, we have
	%	\[
	$\displaystyle{\frac{\X_n -\IE_{\T_n}(\X_n)}{\sqrt{\IV_{\T_n}(\X_n)}} \xrightarrow[]{~d~} \N(0,1).}$
	%	\]
\end{theorem}
The proof is obtained via an adaptation of Flajolet and Sedgewick's approach
for the limit Gaussian law of Stirling Cycle numbers~\cite[p.~644]{FS09}:
see Appendix~\ref{app:high_inc}.
Observing the mean value $\IE_{\T_n}(\X_n)$ we remark that only the second order in the asymptotic 
behavior permits to conclude that some internal nodes are not binary.

\newpage

\subsubsection*{Quantitative characteristics of the root node}

%$\;$

\begin{wrapfigure}[11]{r}{5.8cm}
	\vspace*{-4mm}
	\begin{center}
		\begin{tabular}{ l l l l l l l l  }
			$1$ ,&  $0$, & & & & &  \\
			$0$ ,&  $0$, & $1$, & & & & \\
			$0$ ,&  $0$, & $2$, & $1$,  & & &\\
			$0$ ,&  $0$, & $8$, & $3$,  &  $1$, & & \\
			$0$ ,&  $0$, & $40$, & $15$,  &  $4$, &  $1$,  & \\
			$0$ ,&  $0$, & $240$, & $90$,  &  $24$, &  $5$,  & $1$, \\
			%		$0$ ,&  $0$, & $1680$, & $630$,  &  $168$, &  $35$,  & $6$,  & $1$,  & & & \\
			%		$0$ ,&  $0$, & $13440$, & $ 5040$,  &  $1344$, &  $280$,  & $48$,  & $7$,  & $1$, & & \\		
			%		$0$ ,&  $0$, & $120960$, & $ 45360$,  &  $12096$, &  $2520$,  & $432$,  & $63$,  & $8$, & $1$, & \\		
			%		$0$ ,&  $0$, & $1209600$, & $ 453600$,  &  $120960$, &  $25200$,  & $4320$,  & $630$,  & $80$, & $9$, & $1$, \\
		\end{tabular}
		\caption{Distribution of $t_{n,k}$ for size-$n$ trees, $n \in \{1,2,\dots, 6\}$, of the root arity
			$k \in \{0,1,\dots,n\}$}
	\end{center}
\end{wrapfigure}
The next parameters we are interested in are related
to the root of the increasing Schröder trees. Concerning this particular
node, we want to understand first its typical arity, and then
the number of leaves attached to it in a large tree.

To avoid the description of several notations in the paper,
we have chosen to reuse the previous notations for this new sequence.
Thus here the variable $\U$
marks the arity of the root. The specification is direct: either 
the root-leaf is modified in the evolution process, or it is not the root that is substituted.
%\begin{equation*}%\label{eq:specroot_T2}
\begin{align*}
\T = &\Z \cup \left(\U \times \Theta_\Z (\Z) \times \SeqUn(\U \times \Z) \right)\\
&\cup \left(\Theta_\Z (\T \setminus \Z) \times \SeqUn(\Z) \right).
\end{align*}
%\end{equation*}
We directly obtain the translation
%\begin{equation*}
\[
T(z,u) = z + \frac{u^2\; z^2}{1-uz} + \frac{z^2}{1-z} \; \partial_z \left(T(z,u) - z\right).
\]
%\end{equation*}
%We want to characterize the law of probability of the root degree of a tree $\mathcal{T}^{[2]}_n$. We start this study by defining a recurrence $t_{n,k} = t^{[2]}_{n,k}$ where $t^{[2]}_{n,k}$ is the number of trees of  $\mathcal{T}^{[2]}$ having $n$ leaves and root degree $k$.\\
%
%We have $\forall n \geq 2, t_{n,0}=t_{n,1}=0$ since there are no trees of size $\geq 2$ having root degree of $1$ or $1$ and $\forall n \geq 2, t_{n,n}=1$ since the only tree of size $n$ having a root degree $=n$ is the tree with a single root and $n$ leaves.\\
%
%\begin{proposition} \[ \forall n\geq 3, 2 \leq k \leq n-1, t_{n,k}=\sum\limits_{i=k}^{n-1} i \times t_{i,k}  \]
%\end{proposition}
%
%\begin{proof}
%	A tree of size $n$ and root degree $k$ can come from any tree of smaller size $i < n$ with the same root degree $k$. each of the tree $t_{i,k}$ can be expanded in $i$ different ways corresponding to the expansion of one its leaves 
%	
%\end{proof}
In the same way as before we prove 
%{\small
\begin{align} \label{eq:rec3}
&\left\{ \begin{array}{l}
t_1(u) = 1, \qquad t_2(u) = u^2,  \\
\text{and if } n> 2,\\
t_{n}(u) = u^{n-1}(u-1) + n\; t_{n-1}(u);
\end{array}
\right. \nonumber\\
%\hspace*{0.6cm}
&\left\{ \begin{array}{l c l}
t_{n,k} = n\; t_{n-1,k} & & \text{ if } 1 < k < n-1\\
t_{1,0} = 1, \quad t_{2,2} = 1 & & \\
t_{n,n-1} = n-1 \quad t_{n,n}=1 & &\text{ if } n>2 \text{ and} \\
t_{i,j} = 0 & & \text{ otherwise}.
\end{array}
\right.
\end{align}
%}

\noindent These sequences are related to \href{https://oeis.org/A094112}{\texttt{OEIS A094112,A092582}},
that define properties on permutations (either some avoiding pattern,
or with some fixed size initial run).

\begin{corollary}
	For $n\geq 2$ and $2 \leq k \leq n-1$, we get 
	$\displaystyle{t_{n,k}=n!  \frac{k}{(k+1)!}}$.
\end{corollary} 
A proof by induction is direct.
%
%\begin{proof}
%	By induction, 
%	
%	\[ \begin{array}{l c l}
%	t_{n,k} &= &\sum\limits_{i=k}^{n-1} i \times t_{i,k} \\
%	& = &( \sum\limits_{i=k+1}^{n-1} i \times i \frac{k}{(k+1)!} ) + k  \\
%	& = &( \frac{k}{(k+1)!} \sum\limits_{i=k+1}^{n-1} i \times i  ) + k  \\
%	& = & \frac{k}{(k+1)!} (n!-1-\sum\limits_{i=1}^{k} i \times i  ) + k  \\
%	& = & n!\frac{k}{(k+1)!} (-1-\sum\limits_{i=1}^{k} i \times i  )\frac{k}{(k+1)!} + k  \\
%	& = & n!\frac{k}{(k+1)!}  -\frac{k}{(k+1)!}(1+\sum\limits_{i=1}^{k} i \times i  ) + k  \\
%	& = & n!\frac{k}{(k+1)!}  -\frac{k}{(k+1)!}(k+1)!  + k  \\
%	
%	& = & n!\frac{k}{(k+1)!} \\
%	
%	\end{array} \]
%	
%\end{proof}

\begin{theorem}\label{theo:law_root_arity}
	Let $\X_n$ be the random variable describing the distribution of
	the number of children of the root in increasing Schröder trees of size $n$,
	we have, for $n\geq 2$ and $2 \leq k \leq n-1$,
	\[
	\IP_{\T_n}(\X_n=k) = \frac{2k}{(k+1)!}, \quad \text{and} \quad
	\IP_{\T_n}(\X_n=n) = \frac{2}{n!}.
	\]
\end{theorem}

The second characteristics for the root node is the number of
leaves that are attached to it. Here the specification and thus
the ordinary differential equation are more involved.
In particular, the operators needed for the specification are not so classical
so we prefer to explain directly the differential equation.
Once again, let $T(z,u) = \sum_{n,k} t_{n,k} \; u^k z^n$
be the bivariate generating with $t_{n,k}$ the number
of size-$n$ increasing Schröder trees with $k$ leaves as
children of the root. Then,
{\small
	\[
	T\left( z,u \right) = z + \frac {u^2\; z^2}{1-uz} + \frac{z^2}{1-z}\; \partial_z T\left( z,u \right)
	+ \frac{ z\left( 1-u \right)}{1-z} \; \partial_u T\left( z,u \right). 
	\]
}
Let us give the details to understand the construction. A tree is either reduced to a leaf or 
a single internal node with some leaves: $z + \frac {u^2\; z^2}{1-uz}$.
Or in the iterative process a leaf attached to the root is selected, then replaced by an
internal node with at least two leaves (that are not anymore attached to the root of the whole tree):
$\frac{z}{1-z} \; \partial_u T\left( z,u \right)$.
Or, during the iterative process, a leaf that is not attached to the root is selected and replaced
by an internal node attached to at least two leaves: $\frac{z^2}{1-z}\; \partial_z T\left( z,u \right)
- \frac{ u\; z}{1-z} \; \partial_u T\left( z,u \right)$. The second term removes the trees built
in the first one where we have selected a leaf attached to the root (and also marked by $z$).
%The bivariate generating function $T(z,u)$ is clearly holonomic.
%It was also the case for the previous parameters, but this property was not
%used to derive the previous results. 

\noindent Again by denoting $T(z,u) = \sum_n t_n(u) z^n$,
we can extract the following recurrence, for all $n\geq 4$,
\begin{align*}
t_n(u) = &(n+u) \; t_{n-1}(u) + u(1-n) \; t_{n-2}(u) \\
&+ (1-u) \; t'_{n-1}(u) + (u^2-u) \; t'_{n-2}(u), 
\end{align*}
with $t_1(u)=1$, $t_2(u)=u^2$ and $t_3(u)=u^3+2u$.

%In order to compute the average number of leaves attached to the root, we are interested
%in the cumulative generating series $C(z)= \partial_u T(z,u)_{\mid_{u=1}}$.
%By differentiating the previous recurrence, we directly deduce a recurrence for the coefficients of $C(z)$,
%namely: $c_n = n\; c_{n-1} + (2-n)\; c_{n-2}$ which can be simplified into
%$c_n = 1 + (n-1)\; c_{n-1}$ with $c_1=1$.
%This series is again a classical one,
%registered as
%\href{https://oeis.org/A000522}{\texttt{OEIS A000522}},
%and satisfying the exact formula,
%for $n>1$ by $c_n=\lfloor e\; (n-1)!\rfloor$.
%Thus, the mean number of leaves attached to the root
%in trees of size $n>1$
%is exactly $\displaystyle{\IE_{\T_n}(\X_n) = \frac{2\lfloor e\; (n-1)! \rfloor}{n!}}$.
%
%By the same approach, we get that the coefficients
%of $D(z)=\partial^2_u T(z,u)_{\mid_{u=1}}$ are satisfying 
%$d_{n+1}=(n-1)\; d_n + 2n$ for $n\geq 4$ and $d_1=0$, $d_2=2$ and $d_3=6$.
%We thus get that the second factorial moment verifies 
%$\IE_{\T_n}(\X_n(\X_n-1)) = \frac{\lfloor 4e\; (n-2)! -2 \rfloor}{n!}$ for $n\geq 6$.
%Since $\IV_{\T_n}(\X_n) = \IE_{\T_n}(\X_n(\X_n-1)) - \IE_{\T_n}(\X_n)(\IE_{\T_n}(\X_n)-1)$,
%we obtain the following theorem.
\begin{theorem}
	\label{theo:leaves_root}
	Asymptotically, the mean and the variance of the number $\X_n$ of leaves attached to the root are
	$\IE_{\T_n}(\X_n)=\frac{2e}{n} + O\left(\frac{1}{n!}\right)$ and 
	$\IV_{\T_n}(\X_n)=\frac{2e}{n} + O\left(\frac{1}{n^2}\right)$.
\end{theorem}
%The proof is presented in Appendix~\ref{app:high_inc}.
Let us remark the second term in the expansion of $\IE_{\T_n}(\X_n)$ is
extremely small in front of the main term. % $2e/n$.
%
%\begin{remark}
%	Using guessing techniques, we can observe that the sequence $f_n(u)$ follows the very unexpected (minimal) linear recurrence:
%	\begin{tiny}
%		\begin{align*}
%			&-{u}^{2} \left( u-1 \right) ^{2} \left( n-1 \right) 
%			\left( n+3 \right)  \left( n+1 \right)f_n \\
%			&+3\, \left(  \left( {n}^{2}+3\,n
%			+1 \right) {u}^{2}+ \left( {n}^{3}+4/3\,{n}^{2}-{\frac {41\,n}{3}}-{
%				\frac{71}{3}} \right) u-1/3\, \left( n+4 \right)  \left( n-7 \right) 
%			\left( n+2 \right)  \right)  \left( u-1 \right) uf_{n+1}\\
%			&+(  \left( -3n-6 \right) {u}^{4}+ \left( -9n^{2}-17n+34
%			\right) u^3\\
%			&\quad + \left( -3n^3+13n^2+134n+187 \right) u^2+ \left( 2n^3-13n^2-155n-284 \right) u\\
%			&\quad+5n^2+
%			40n+75) f_{n+2}\\
%			&+ \left( {u}^{4}+ \left( 9\,n+14 \right) {u}^{3}+ \left( 9\,{
%				n}^{2}-5\,n-128 \right) {u}^{2}+ \left( {n}^{3}-16\,{n}^{2}-110\,n-120
%			\right) u+10\,{n}^{2}+80\,n+155 \right) f_{n+3} \\
%			& +
%			\left( -3\,{u}^{3}+ \left( -9\,n-4 \right) {u}^{2}+ \left( -3\,{n}^{2
%			}+23\,n+131 \right) u+5\,{n}^{2}+25\,n+5 \right) f_{n+4}\\
%			& + \left( 3\,{u}^{2}+ \left( 3\,n-9 \right) u-10\,n-40 \right) f_{n+5} = \left( u-5 \right) f_{n+6}
%		\end{align*}
%	\end{tiny}
%	and this recurrence can be use alternatively to reach all the moment of the distribution.
%\end{remark}

\subsubsection*{Quantitative analysis of the number of binary nodes}

Here the specification is easier to exhibit, 
and its translation via the symbolic method is direct ($\U$ is marking the binary nodes):
\begin{align*}
&\T = \Z \cup \left(\Theta_{\Z} \T \times \left( \U \times \Z \cup \SeqDeux (\Z) \right) \right);\\
%\hspace*{1.5cm}
&T(z,u) = z + \left( u\; z^2 + \frac{z^3}{1-z} \right) \partial_z T(z,u).
\end{align*}
Let us again extract the recurrence $t_n(u)$, for all $n\geq 4$:
\[
t_n(u)= (1+u(n-1)) t_{n-1}(u) + (1-u)(n-2) t_{n-2},
\]
with $t_1(u)=1, t_2(u)=u$ and  $t_3(u)=1+2u^2$. 
Note that due to this recurrence,
the probability distribution $p_n(u)=\frac{2}{n!} t_n(u)$
also exhibits a simple recurrence
(cf.~\cite[p. 157]{FS09}).
Thus we easily compute the mean and the second factorial moment
of the number of binary nodes in size-$(n\geq 3)$ trees:
\begin{align*}
&\IE_{\T_n}(\X_n) = \frac{7}{3} + n - 2\cdot \sum_{k=1}^{n} \frac{1}{k} - \frac{1}{n},\quad \text{and} \\
&\IE_{\T_n}(\X_n(\X_n-1)) = \sum _{k=3}^{n} \frac{2k-2}{k} \IE_{\T_{k-1}}(\X_{k-1}) \\
& \hspace*{3cm} - \frac{2k-4}{k (k-1) } \IE_{\T_{k-2}}(\X_{k-2}).
\end{align*}

\begin{theorem}
	Asymptotically, the mean and the variance of the number $\X_n$ of binary internal nodes
	are
	$\displaystyle{\IE_{\T_n}(\X_n) = n -2\ln(n) + \frac{7}{3} -2\gamma + O\left(\frac{1}{n}\right)}$ and 
	$\displaystyle{\IV_{\T_n}(\X_n) = 4\ln(n)+O(1)}$. Furthermore there is a limiting distribution satisfying
	$\displaystyle{\frac{\X_n -\IE_{\T_n}(\X_n)}{\sqrt{\IV_{\T_n}(\X_n)}} \xrightarrow[]{~d~} \N(0,1).}$
\end{theorem}
\begin{proof}[Key-ideas]
	The recurrences give the closed form formulas
	for $\IE_{\T_n}(\X_n)$ and $\IV_{\T_n}(\X_n)$.
	Then it is important to notice that $t_n(u)$ can be approximated by $\tilde{t}_n(u)$ verifying
	$\tilde{t}_n(u)=(1+u(n-1))\tilde{t}_{n-1}(u)$. The latter recurrence
	is the same recurrence as the one exhibited for the number of internal nodes
	(equation~\eqref{eq:rec2}).
	Thus, by the same arguments, the sequence of distributions $(t_n(u))$
	converges in distribution to a Gaussian law.
\end{proof}

\subsection{Bijection with permutations}
\label{sec:permut}

Observing the exact value $t_n = n! / 2$ enhances the chances of finding some relation
between our model of increasing trees and a subclass of permutations.
Let us start with this goal.
First, for a size-$n$ permutation $\sigma$ denoted by $(\sigma_1, \dots, \sigma_n)$,
we define $\sigma_i$ to be its $i$-th element (the image of $i$),
and $\sigma^{-1}(k)$ to be the preimage of~$k$ (the position of $k$ in the permutation).
We are now ready to define the recursive map $\M$ between $\mathcal{HP}$,
the class of permutations such that $1$ appears before $2$  and the class $\T$
of increasing Schröder trees.
The base case is the permutation $(1,2)$ which corresponds to the root labeled by~$1$
attached to two unlabeled leaves. Let $\sigma$ be a size-$n$ permutation in $\mathcal{HP}$,
with $n\ge 3$. We observe its the greatest element: if $\sigma_n=n$ then we add a new rightmost
leaf to the last added internal node (the one with the largest label);
otherwise let $k=\sigma^{-1}(n)$, we create a new binary node $\nu$ labeled with a new integer
(the smallest as possible) and attached to two new leaves, then we replace the $k$-th
leaf by this new tree rooted at $\nu$, in the tree under construction based on $\sigma_{\backslash n}$,
i.e. that is $\sigma$ without the greatest element $n$. Remark that during
the tree construction we must traverse the leaves, we can take an arbitrary traversal.

\begin{theorem}
	The map $\M$ is a one-to-one correspondence between $\mathcal{HP}$ and $\T$.
\end{theorem}
\begin{proof}
	The mapping is size preserving: at each iteration we remove exactly 
	one element from the permutation and add exactly one leaf to the tree by
	either adding a leaf to the last exiting node or by killing one leaf and
	adding to new ones.
	The mapping is injective since by induction at each
	iteration we remove the greatest element of the permutation and its
	following its index the actions are performed on the resulting tree
	in a non-ambiguous manner.
	Finally the mapping is based on two classes with the same number
	of elements (of each size).
\end{proof}

\begin{figure}[ht]
	
	\begin{tabular}{l|l} 
		\begin{minipage}{1cm}
			\vspace*{-8mm}
			$(1,2)$
		\end{minipage} 
		\begin{minipage}{1cm}
			\vspace*{-8mm}
			$\xrightarrow{\M}$
		\end{minipage} 
		\includegraphics[scale=0.3]{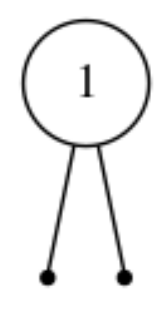}
		&  
		\begin{minipage}{15mm}
			\vspace*{-8mm}
			$(1,2,3)$
		\end{minipage} 
		\begin{minipage}{10mm}
			\vspace*{-8mm}
			$\xrightarrow{\M}$
		\end{minipage}  
		\includegraphics[scale=0.3]{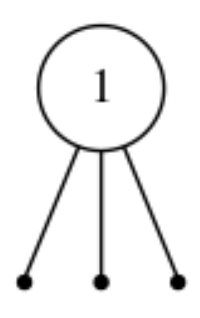}
		\\\hline
		
		\hspace*{-3mm}
		\begin{tabular}{lr}
			\begin{minipage}{18mm}
				\vspace*{2mm}
				$(4,1,2,3)$\\
				
				\hfill{$\xrightarrow{\M}$}
				\vspace*{1cm}
				$\;$
			\end{minipage} & 
			\begin{minipage}{9mm}
				%\phantom{XLM} 
				$\;$ \\
				\includegraphics[scale=0.3]{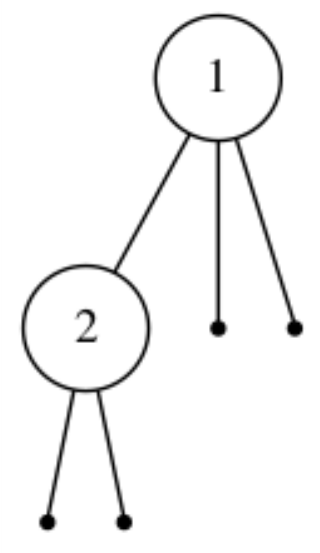}
			\end{minipage}
		\end{tabular}
		&
		\hspace*{-3mm}
		\begin{tabular}{lr}
			\begin{minipage}{18mm}
				\vspace*{1mm}
				$(4,1,2,5,3)$\\
				
				\hfill{$\xrightarrow{\M}$}
				\vspace*{1cm}
				$\;$
			\end{minipage} & 
			\begin{minipage}{11mm}
				% \phantom{XLM}
				$\;$ \\
				\includegraphics[scale=0.3]{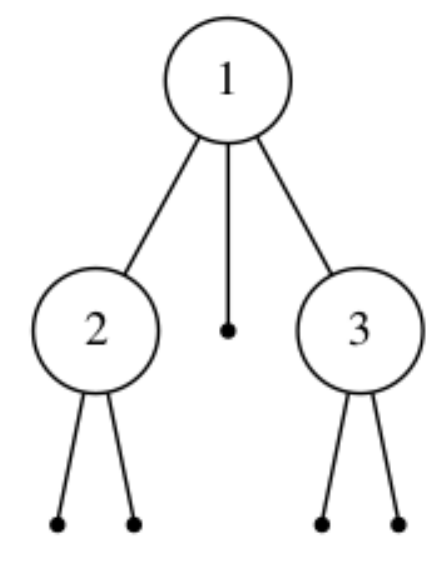}
			\end{minipage}
		\end{tabular}
		\\ \hline
		
		\hspace*{-3mm}
		\begin{tabular}{lr}
			\begin{minipage}{18mm}
				\vspace*{2mm}
				$(4,1,2,5,3,6)$\\
				
				\hfill{$\xrightarrow{\M}$}
				\vspace*{1cm}
				$\;$
			\end{minipage} & 
			\begin{minipage}{12mm}
				$\;$ \\
				\includegraphics[scale=0.3]{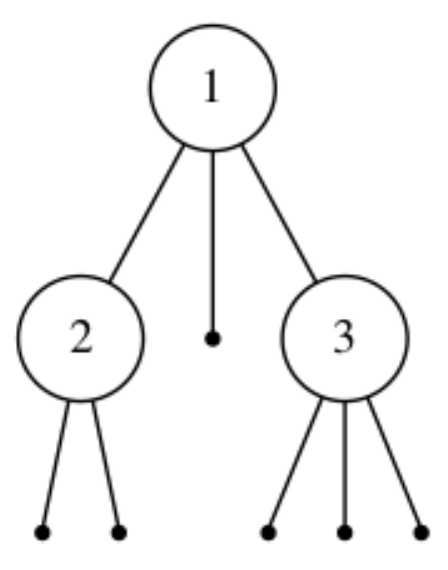}
			\end{minipage}
		\end{tabular}
		&
		\hspace*{-3mm}
		\begin{tabular}{lr}
			\begin{minipage}{20mm}
				\vspace*{1mm}
				$(4,1,2,5,3,6,7)$\\
				
				\hfill{$\xrightarrow{\M}$}
				\vspace*{1cm}
				$\;$
			\end{minipage} & 
			\begin{minipage}{12mm}
				$\;$ \\
				\includegraphics[scale=0.3]{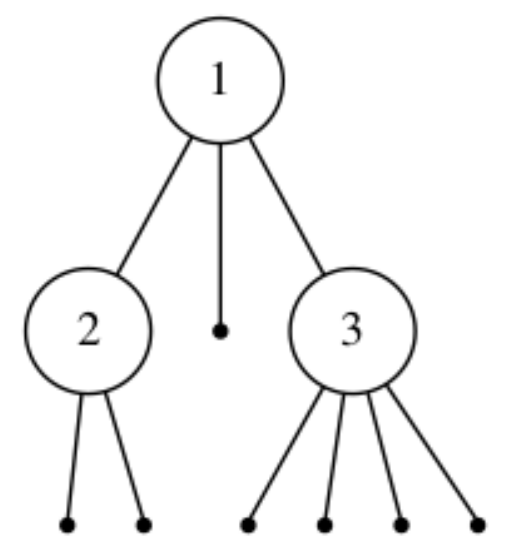}
			\end{minipage}
		\end{tabular}
		\\
		\hline
	\end{tabular} 
	
	\begin{center}
		\begin{minipage}{25mm}
			$(4,1,2,5,3,8,6,7)$
			\vspace*{15mm}
			$\;$
		\end{minipage} 
		\begin{minipage}{1cm}
			\quad   $\xrightarrow{\M}$
			\vspace*{20mm}
			$\;$
		\end{minipage} 
		\includegraphics[scale=0.3]{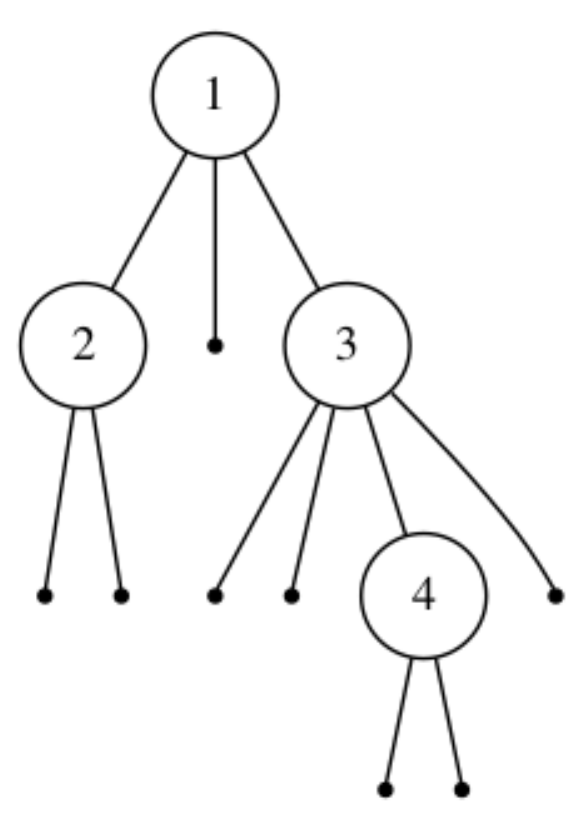}
	\end{center}
	\vspace*{-1.2cm}
	%\hline
	\hrulefill
	\vspace*{2mm}
	
	\caption{A size-$8$ example of the mapping $\M$}
	\label{fig:mapping_ex}
	
\end{figure}
In the Figure~\ref{fig:mapping_ex} we present the mapping
on an example. Remark that we have ordered the steps reversely
to understand the process in a easiest way.

\subsection{Uniform random sampling}
\label{sec:random}

Obviously, through the latter bijection we are able to obtain an uniform random sampler.
It suffices to uniformly sample permutations and to use the bijection to build the
associated increasing Schröder tree.
While there exists fast algorithms to sample 
permutations, see for example~\cite{BBHT17},
using the bijection efficiently is not obvious.

However, through the bijection a direct \textit{probabilistic construction} of increasing Schröder
trees can be obtained.
Such a probabilistic construction presents two main advantages.
Firstly, it simplifies the implementation of a sampler and, secondly,
and more importantly,
it gives a purely probabilistic approach to the original combinatorial class,
which we are interesting in.
This probabilistic approach can be used to compare to other probabilistic tree models
or to exhibit important characteristics of trees on average
using probabilities rather than combinatorics.

We introduce in this section an algorithm to uniformly sample increasing Schröder trees
of a given size~$n$. 
A first remark is that the uniform sampling of structures with increasing labeling
constraints is not so classical in the context of Analytic Combinatorics.
There are some studies by Mart{\'i}nez and Molinero~\cite{MM03,Molinero05}
in the context of the recursive method and
some other about Boltzmann sampling either directly for the method~\cite{BRS12}
or focusing on a specific application~\cite{BDFGH16}.

For the uniform sampling of our evolution process, we are focusing on two goals.
Our fundamental goal consists in controlling the probability distribution used for the sampling.
In fact, we may extract some statistical information based on the samplings, thus the probability
distribution is central. We choose to sample uniformly trees of the same size,
because then we can bias our generator (and tune the bias) to construct other probability distributions.
Secondly our algorithmic framework
must be very efficient to sample large trees (with several thousands of leaves).
Thus a detailed complexity analysis is necessary to be sure that the algorithm cannot be easily
improved.

Our approach is based on the combinatorics underlying the very efficient recurrence
$t_n = n\cdot t_{n-1}$:
a tree of size $n$ can be built from a tree of size $n-1$ in $n$ different ways.
We exhibit a construction based on this recurrence.
%Having an $(n-1)$-permutation, we have $n$ different places for adding a new element
%to build an $n$-permutation (like the Fisher-Yates algorithm).
%Recall that we take canonical permutations in the sense that a size $k$ permutation has elements $1$ to $k$.
%Therefore from an $(n-1)$-permutation to an $n$-permutation, we just add the integer $n$ somewhere in the permutation.
%Due to our bijection if we have a permutation $\sigma$ with of size $n-1$,
%and the next element is added at the last place of the permutation
%then $m(\sigma')=m(\sigma)+1$ thus we only add a new leaf to the last node.
%If the element is not added at the last place of the permutation
%then it will create a new binary node (cf. Section~\ref{sec:permut}).
This leads to the following iterative algorithm of random sampling.
\begin{algorithm}[h]
	\caption{\label{algo:stronglyTreeBuilder} Increasing Schröder Tree Builder}
	\begin{algorithmic}[1]
		\Function{TreeBuilder}{$n$}
		\If{$n=1$}
		\State \textbf{return} the single leaf
		\EndIf
		\State $T :=$ the root labeled by $1$ and attached to two leaves
		\State $\ell := 2$ 
		\For{\texttt{$i$ from $3$ to $n$}}
		\State $k := rand\_int(1,i)$
		\If{$k=i$}
		\State Add a new leaf to the last added internal node in~$T$
		\Else
		\State Create a new binary node at position $k$ in $T$
		\State \phantom{Create} with label $\ell$ and attached to two leaves
		\State $\ell :=\ell+1$ 
		\EndIf
		\EndFor
		\State \textbf{return} T
		\EndFunction
	\end{algorithmic}
	{\footnotesize
		The function $rand\_int(a,b)$ returns
		uniformly at random an integer in $\{a, a+1, \dots, b\}$.\\
	}
\end{algorithm}
%
%\begin{corollary}
%	(Correctness) The algorithm sample trees uniformly in all trees of the class $\T$
%\end{corollary}
%The proof is a direct consequence of the mapping $\M$. 
%This algorithm gives the probabilistic construction of trees of $\T$. 
%\begin{theorem}
%	The function \textsc{TreeBuilder($n$)} in Algorithm~\ref{algo:stronglyTreeBuilder} necessitates $\Theta(n \log n)$ random bits.
%\end{theorem}
%The number of random bits needed is $\sum\limits_{i=3}^n \log i=\Theta(n \log n)$.
%\begin{theorem}
%	The function \textsc{TreeBuilder($n$)} in Algorithm~\ref{algo:stronglyTreeBuilder}
%	is a uniform sampling algorithm for size-$n$ trees.
%	It operates in $O(n)$ operations on trees.
%\end{theorem}

\begin{theorem}
	The function \textsc{TreeBuilder($n$)} in Algorithm~\ref{algo:stronglyTreeBuilder}
	is a uniform sampling algorithm for size-$n$ trees.
	Asymptotically, it operates in $O(n)$ operations on trees
	and necessitates $O(n \ln n)$ random bits.
\end{theorem}
The correctness of the algorithm is a direct consequence of the mapping $\M$. 
it gives the probabilistic construction of trees of $\T$.
Using the adequate data structures, as for example by keeping an array of pointers
to all leaves and another one to the last inserted internal node,
each insertion in the tree under construction is done in constant time.

\section{Weakly Increasing Schröder trees}
\label{sec:weak_inc}

In this section we aim at developing another model for ranked trees
based on Schröder structures. In fact we relax somehow 
the labeling constraint.

\subsection{The model and its context}

Weakly increasing Schröder trees are a generalization of strongly increasing Schröder trees.
The tree structure is still an unlabeled Schröder tree. But the labeling is different.
Internal nodes are labeled between $1$ to $\ell$ in such a way that the sequence
of labels in each path from the root to a leaf is also increasing.
The difference here is that different nodes can have the same label.
This model is also built iteratively.
\begin{itemize}
	\item \emph{Start with a single (unlabeled) leaf;}
	\item Iterate the following process: \emph{at step $\ell$ (for $\ell\geq 1$),
		select a subset of leaves and replace each of them
		by an internal node with label $\ell$
		attached to a sequence of at least two leaves.		} 
\end{itemize}
\begin{figure}[ht]
	\begin{center}
		\includegraphics[scale=0.4]{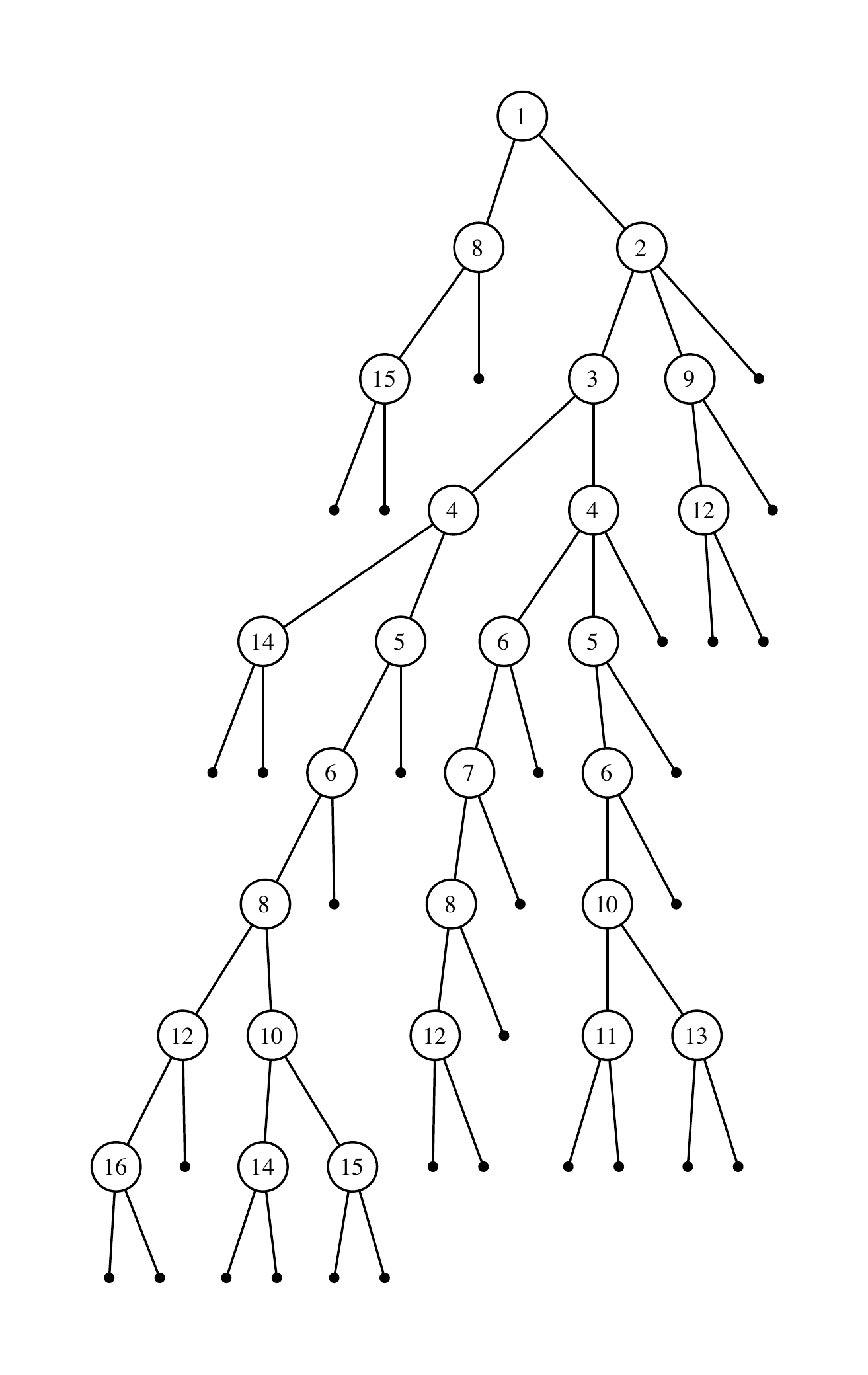}
		\caption{\label{fig:weak_inc_tree_ini}
			A weakly increasing Schröder tree}
	\end{center}
\end{figure}
In Figure~\ref{fig:weak_inc_tree_ini} we present a weakly increasing tree of size~$30$
with $16$ distinct labels..

\subsection{Exact enumeration and relationship with ordered Bell numbers}
\label{sec:enum2}

We can specify the process through the symbolic method.
But once again the labeling is transparent and does not appear in the specification.
\begin{equation}
\label{eq2:weak_funeq}
G(z) = z+ G\left(\frac{z^2}{1-z}+z\right) - G(z).
\end{equation}
At each iteration and for each leaf we can either leave it as it is 
or expand it into a new internal node with at least 2 leaves. 
The configuration where no leaf is expanded is forbidden, thus we remove $G(z)$
in equation~\eqref{eq2:weak_funeq}.
From this equation we extract the recurrence
\begin{equation}
\label{eq:rec_weak}
g_n = \left\{ \begin{array}{l c l}
1 & \quad & \text{if } n=1 \\
\sum\limits_{k=1}^{n-1}\binom{n-1}{k-1} \; g_k. & & \text{otherwise.}
\end{array}
\right.
\end{equation}
The first coefficients correspond to a shift
of the sequence of Ordered Bell numbers 
(also called Fubini numbers) referenced as
\href{https://oeis.org/A000670}{\texttt{OEIS A000670}}.
{\small
	\[g_n = 0,1, 1, 3, 13, 75, 541, 4683, 47293, 545835,  7087261,\dots \]
}
By following the approach developed by Pippenger in~\cite{Pippenger10}
for the derivation of the exponential generating function for ordered Bell numbers
we obtain, by starting from our equation~\eqref{eq2:weak_funeq},
$(\B G)'(z) = 1/(2-e^z)$. Thus, after integration
$ \displaystyle{\B G(z) = \frac{1}{2}\left( z - \ln{(2-e^z)}\right)}$.
Usually ordered Bell numbers are specified by $B=\Seq(\Set_{\geq 1}(z))$.
Obviously this gives the exponential generating function $B(z)= 1/ (2-e^z)$.
Thus, we have proved that our sequence is a shift of the one of ordered Bell numbers.
As a by-product, we have exhibited a new way for specifying ordered Bell numbers.

Recall the $n$-th ordered Bell number, denoted by $B_n$,
counts the total number of partitions of a set of size $n$
where additionally we consider an order over the subsets of the partition.
\begin{equation}
B_n = \sum\limits_{k=0}^n k! \; \stirling{n}{k} \sim \frac{n!}{2 \left(\ln 2\right)^{n+1}},
\end{equation}
where $\stirling{n}{k}$ stands for the 
Stirling partition numbers (also called Stirling numbers of the second kind).
The number $B_n$ corresponds to the number $g_{n+1}$ of 
weakly increasing Schröder trees of size $n+1$.

\subsection{Bijection between ordered Bell numbers and weakly increasing
	Schröder trees}
\label{sec:Bell}

In ordered partitions, the subsets are ordered but the elements inside a subset are not.
In the following let us denote by $p=[p_1,p_2,\dots,p_\ell]$ an ordered partition 
such that $p_i$ is the subset of the partition at position $i$.
For example if $p=[ \{3,4 \} , \{1,5,7 \}, \{ 2,6\} ] $, then $p_1=\{3,4 \}$. %, $p_3= \{ 2,6\}$.
We denote by $|p_i|$ the size of the $i$-th subset: $|p_1| = 2$.
The total size (i.e. number of elements) of the partition is denoted by $|p|$
Thus the elements of an ordered partition range from $1$ to $|p|$. 

For the exhibition of the correspondence we will use a
canonical order inside the subsets, 
consisting in enumerating the elements increasingly.

Let $p = [p_1, p_2,\dots, p_\ell]$ be an ordered partition,
and $p_i = \{\alpha_1, \alpha_2, \dots, \alpha_r\}$ (with $r\geq 1$),
such that  $\alpha_1< \alpha_2< \dots < \alpha_r$. 
We define a \emph{run} in $p_i$ to be a maximal sequence $\alpha_i, \alpha_{i+1},\dots, \alpha_j$
equal to $\alpha_i, \alpha_{i}+1,\dots, \alpha_i+j-i$.
It is maximal in the sense that $\alpha_{i-1} < \alpha_i - 1$
and $\alpha_{j+1} > \alpha_j+1$.
We define the map $runs$ that lists all the runs of a subset.

For instance, in our example $p$,
in $p_1$ there is a single run: $3,4$ and in $p_2$, there are
$3$ runs.

The mapping deals with incomplete ordered partitions
(in the sense that some integers are not present in the partition).
We define a normalization of a partition, denoted by $norm$,
that maps an incomplete ordered partition of size $k$
into the corresponding ordered partition of size $k$
whose elements are $\{1,\dots,k\}$ and that keeps the relative order between
the elements.
For example by taking the first two subsets from $p$ as $p'=[ p_1,p_2 ]$,
then $p'$ is an incomplete ordered partition of size $5$
and we get $norm(p')=[ \{2,3 \} , \{1,4,5 \} ]$. 

From the ordered partition $p$, the mapping $\M'$
builds the corresponding tree by processing the subsets of the ordered partition successively. 
We start by creating a new ordered partition that contains only $p_1$, $p'=norm([p_1])$.
The size of $p_1$ determines the arity of the root: it equals $|p_1|+1$.
The root label is $1$. Then at each step $i$ with $i \in \{2,\dots,\ell \}$,
we process the subset $p_i$ as follows. 
Normalize the incomplete ordered partition  $[p_1, p_2,\dots, p_i]$.
In the normalized ordered partition the corresponding subset of $p_i$
is denoted by $p'_i$.
The number of new internal nodes is $|runs(p'_i)|$, all labeled by~$i$.
Suppose $runs(p'_i)=[r_1,r_2,\dots,r_j]$ (with each $r_\ell$ a 
set of successive integers and possibly a single one).
Take an order for the leaves in the tree under construction
(the postorder one for example)
and iterate the process:
For $\ell$ from $1$ to $j$, take the leaf whose index is the first element of $r_\ell$
and replace it with an internal node with label $i$ of arity $|r_\ell| +1$.

\begin{figure}[ht]
	\begin{tabular}{c|c} 
		\begin{minipage}{0.22\textwidth}
			\begin{center}
				%{\footnotesize  
				$[\{3,4\}] \cong [\{1,2\}]$ %}
				\\ \vspace*{4mm}
				$\big\downarrow \M'$ \\
				\vspace*{3mm}
				\includegraphics[scale=0.35]{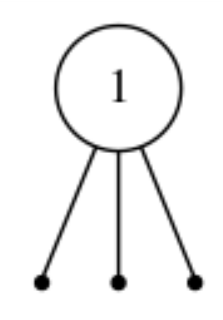}
			\end{center}
			\vspace*{4mm}
			$\;$
		\end{minipage}
		&  
		\begin{minipage}{0.21\textwidth}
			\begin{center}
				%{\footnotesize
				$[\{3,4\},\{1,5,7\}]$ \\
				$\cong  [\{2,3\},\{1,4,5\}]$%}
				\\ \vspace*{4mm}
				$\big\downarrow \M'$\\
				\vspace*{3mm}
				\includegraphics[scale=0.35]{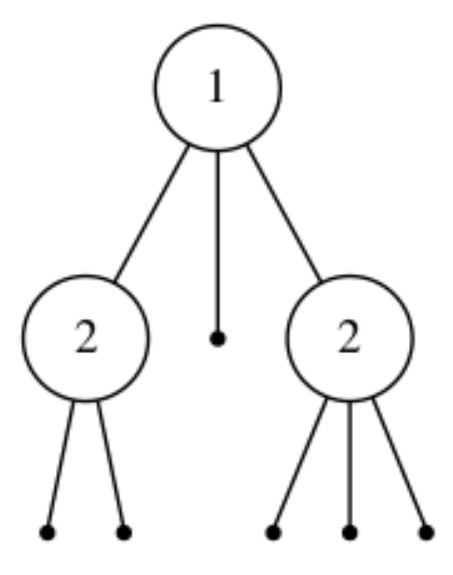}
			\end{center}
		\end{minipage}
	\end{tabular}
	
	\vspace*{-1pt}
	\hrulefill
	\begin{center}
		$[\{3,4\},\{1,5,7\},\{2,6\}] \cong  [\{3,4\},\{1,5,7\},\{ 2,6\}]$
		\\ \vspace*{4mm}
		$\big\downarrow \M'$ \\
		\vspace*{3mm}
		\includegraphics[scale=0.4]{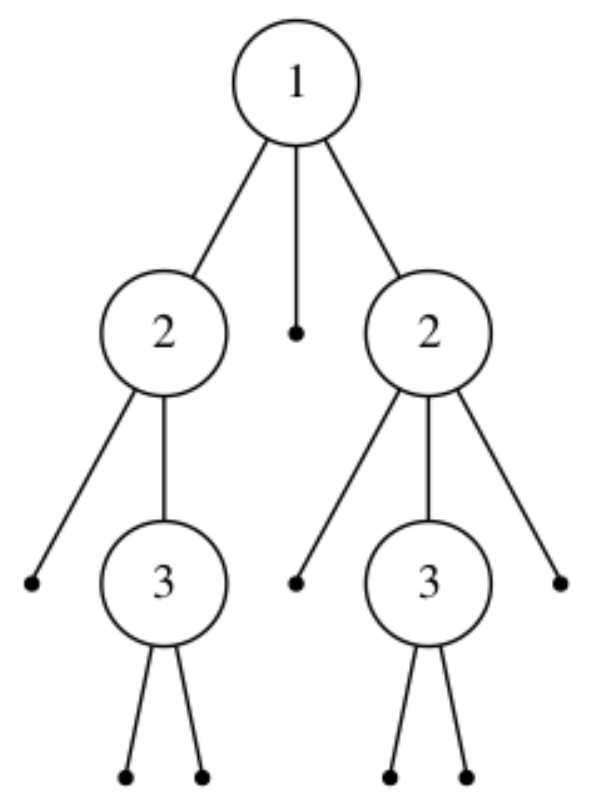}
	\end{center}
	\vspace*{-5mm}
	\hrulefill
	
	\caption{Weakly increasing tree of size $8$}
	\label{fig:bij_weakly}
\end{figure}

In Figure~\ref{fig:bij_weakly} the mapping $\M'$ is applied on our example 
$p=[ \{3,4 \} , \{1,5,7 \}, \{ 2,6\} ] $. The resulting weakly increasing tree
is of size $9$.

\subsection{Analysis of typical parameters}
\label{sec:param2}

\subsubsection*{Quantitative analysis of the number of iteration steps}

In our classical iterative equation, we add a new variable $u$
to mark each iteration.
\begin{equation} \label{eq2:funeq_biv}
G(z,u) = z+ u\; G \left( \frac{z^2}{1-z}+z, u \right) - u \; G(z,u).
\end{equation}
Which leads to the following recurrence,
\begin{equation} \label{weakly:eq:rec}
g_{n,k} = \left\{ \begin{array}{l c l}
1 & & \text{if } n=1, k=0 ,\\
\sum\limits_{j=1}^{n-1}\binom{n-1}{j-1} \; g_{j,k-1} & & \text{otherwise } .\\
\end{array}
\right.
%\hspace*{3cm}
\end{equation}

\begin{wrapfigure}[10]{r}{6.2cm}
	\vspace*{-8mm}
	\begin{center}
		{\footnotesize
			\begin{tabular}{ l l l l l l l l l l l }
				$0$ ,&  & & & & & & & & & \\
				$1$ ,&  & & & & & & & & & \\
				$0$ ,&  $1$, & $2$ , & & & & & & &  &\\
				$0$ ,&  $1$, & $6$, & $6$,  & & & & & &  &\\
				$0$ ,&  $1$, & $14$, & $36$,  &  $24$, & & & & &  &\\
				$0$ ,&  $1$, & $30$, & $150$,  &  $240$, &  $120$,  & & & &  &\\
				$0$ ,&  $1$, & $62$, & $540$,  &  $1560$, &  $1800$,  & $720$  & & &  &\\
			\end{tabular}
		}
		\caption{\label{tab:interne/leaf}
			Distribution of $g_{n,k}$ for $n \in \{0,1,2,\dots, 6\}$}
	\end{center}
\end{wrapfigure}

This recurrence is analogous to the one relating ordered
Bell numbers and Stirling partition numbers.
%In the next section \ref{sec:Bell} we exhibit the exact relation
%between them and a one to one correspondence.
%With this correspondence we obtain results on the law of the number of iteration steps. From the ones of Stirling partition numbers.
%\begin{theorem}
%The average number of iteration steps is a limit Gaussian law. 
%\[\mu \sim c_1 n \quad \sigma \sim \sqrt{c_2 n}\]
%With $c_1 = \frac{1}{2\ln2} \approx 0.72  \quad c_2 = \frac{1-\ln 2}{4(\ln 2)^2} \approx 0.16$ 
%\end{theorem}
\begin{theorem}\label{theo:limit_law_internal_in_weak}
	The distribution of the
	the number of building steps in weakly increasing Schröder trees of size $n$ satisfies
	\[
	g_{n,k} = k! \; \stirling{n+1}{k}.
	\]
	Let $\X_n$ be the random variable describing this distribution,
	we have
	\[
	\frac{\X_n -\IE_{\G_n}(\X_n)}{\sqrt{\IV_{\G_n}(\X_n)}} \xrightarrow[]{~d~} \N(0,1),
	\]
	with $\IE_{\G_n}(\X_n) \sim \frac{1}{2\ln2} \; n$ and $\IV_{\G_n}(\X_n) \sim \frac{1-\ln 2}{(2\ln 2)^2} \; n$.
\end{theorem}
\begin{proof}
	The one-to-one correspondence between weakly increasing Schröder trees and ordered Bell numbers
	gives the combinatorial proof of the distribution for $(g_{n,k})$.
	
	The analysis of the limiting distribution
	is classical in the quasi-powers framework. See for example~\cite[p.~653]{FS09}. 
\end{proof}

%
%Recall.
%Stirling partition numbers (also called Stirling numbers of the second kind) represents the number of ways to partition a set of n objects into k non-empty subsets without considering an order on these subsets. They are defined as follow:
%\[ \stirling{n+1}{k} = k\stirling{n}{k}+\stirling{n}{k-1} \]
%
%If the order of the subsets is to be considered then we multiply by $k!$ which gives
%\[ k! \stirling{n+1}{k} = k ~ (k!\stirling{n}{k}+ (k-1)!\stirling{n}{k-1}) \]
%
%Ordered Bell numbers count the total number of partitions of an $n$-elements set where additionally we consider the subsets of partition ordered between each other. For example $p_1= [ \{ 1,2,4\} , \{3,5 \} ]$ is different from $p_2= [ \{3,5 \}, \{ 1,2,4\}  ]$.\\
%
%\[ \hat{B}_n = \sum\limits_{k=0}^n k!\stirling{n}{k}\]
%
%\begin{theorem}\label{tree:stirling}
%	The complete distribution of the number of 
%\[ g_{n,k} = k! \stirling{n+1}{k}\]
%\end{theorem}
%A direct proof by recurrence for this theorem is possible,
%but through the one-to-one correspondence between weakly increasing Schröder trees
%and ordered Bell numbers (in Section~\ref{sec:Bell})
%a combinatorial proof is obtained.

\subsubsection*{Quantitative analysis of the number of internal nodes}

In this model the number of iteration steps does not correspond
to the number of the internal nodes as at each iteration
any subset of leaves can be expanded into internal nodes with
new leaves.
The specification marking both internal nodes and leaves is
\begin{equation} \label{eq2:funeq}
G(z,u) = z+ G\left(\frac{uz^2}{1-z}+z, u \right) - G(z,u).
\end{equation}
We recall that the substitution $G(\frac{uz^2}{1-z}+z)$ means
that for each iteration each leaf can be left as it is
$z$ or expanded into an internal node of unbounded arity with
new leaves $\frac{z^2}{1-z}$.
It is in the second part that an internal node will be created
and thus we mark it with $u$.
\begin{theorem} \label{weakly:internal nodes}
	The average number of internal nodes $\IE_{\G_n}(\X_n)$
	in size $n$	weakly increasing trees verifies
	\[
	\IE_{\G_n}(\X_n) = n-\ln{n} + o(1).
	\]
\end{theorem}
The main ideas of the proof are in Appendix~\ref{app:weak}.

\subsection{Uniform random sampling}
\label{sec:random2}

We introduce in this section an algorithm to uniformly sample
weakly increasing Schröder trees of a given size $n$
directly, without an intermediate step of generating uniformly an ordered 
partition.

The global approach for our algorithmic framework deals with the
\textit{recursive generation method} adapted to the Analytic Combinatorics
point of view in~\cite{FZVC94}. 
But in our context, we note that we can obtain for free (from a complexity view)
an \emph{unranking algorithm}.
This kind of algorithm has been developed in the 70’s by Nijenhuis and Wilf~\cite{NW75}
and then has been introduced to the context
of Analytic Combinatorics by Mart\'inez and Molinero~\cite{MM03}.
Here the idea is not to draw uniformly an object, but first to define a total order over
the objects under consideration (here weakly increasing Schröder trees)
and then an integer (named the rank) is sampled to build 
deterministically the associated object. Such an approach 
gives also a way to do exhaustive generation
(refer to the paper~\cite{BDGV18} for an example of both methods:
recursive generation and unranking).

For both types of algorithms (recursive generators and unranking ones),
there is a first step of pre-computations (done only once before the sampling
of many objects). We must compute (and store) the numbers of trees
of sizes from $1$ to $n$. Here this phase can be done with a quadratic complexity
(in the number of arithmetic operations) because
of the recursive formula for $g_n$ (cf. equation~\eqref{weakly:eq:rec}).

The second (and last) step for the sampling consists in the recursive construction
of the tree of rank~$r$ that corresponds to an uniformly sampled integer
in $\{0, 1,\dots, g_n-1\}$.
For this purpose, we come back to the original recursive equation~\eqref{eq:rec_weak},
and in particular, we look at the sum over decreasing~$k$:
\[
g_n = \binom{n-1}{n-2} g_{n-1} + \binom{n-1}{n-3} g_{n-2} 
+ \dots + \binom{n-1}{0}g_1.
\]
The latter recurrence is combinatorially easy to understand.
Through the evolution process, to build a size~$n$ tree, 
we take a size $k\in\{1,\dots, n-1\}$ tree constructed
with exactly one less iteration. The binomial coefficient
$\binom{n-1}{k-1}$ corresponds to the number of composition
of $n$ in $k$ parts. Then we traverse 
the tree, and each time we see a leaf, we do the following rule:
if the next part is of value $1$, we leave the leaf unchanged
otherwise for a value $\ell >1$,
we replace the leaf by an internal node (well labeled
with the single new value valid for this step)
and attached to it $\ell$ leaves. We then
take the next part of the composition into consideration
and continue the tree traversal.

In the latter sum, the first term is much bigger than the second one, 
that is must bigger than the third one and so on.
This approach, focusing first on the dominant terms corresponds to the
\emph{Boustrophedonic order} presented in~\cite{FZVC94}.
It allows to improve essentially the average complexity of the random sampling
algorithm. In our case of weakly increasing Schröder trees that
do not follow a standard specification (cf.~\cite{FZVC94}), the complexity gain
is even better. 

\begin{algorithm}[h]
	\caption{\label{algo:TreeBuilder} Weakly increasing Tree Unranking}
	\begin{algorithmic}[1]
		\Function{UnrankTree}{$n, s$}
		\If{$n=1$}
		\State \textbf{return} the tree reduced to a single leaf
		\EndIf
		\State $k := n-1$
		\State $r := s$
		\While{$r >= 0$}
		\State $r := r - \binom{n-1}{k-1}\cdot g_k$
		\State $k := k - 1$
		\EndWhile
		\State $k := k+1$
		\State $r := r + \binom{n-1}{k-1}\cdot g_k$
		\State $s' := r \mod g_k$
		\State $T := $\textsc{UnrankTree}($k,s'$)
		\State $C := $\textsc{UnrankComposition}($n, k, r // g_k$)
		\State Substitute in $T$ some leaves according to $C$ 
		\State \textbf{return} the tree $T$
		\EndFunction
	\end{algorithmic}
	{\footnotesize
		The sequence $(g_k)_{k\leq n}$ and $(\ell!)_{\ell\in\{1,\dots,n\}}$ have been precomputed and stored.\\
		Line $13$: The operation $//$ is the Euclidean division.\\
	}
%	\caption{\label{algo:CompoBuilder} Composition unranking}
	\begin{algorithmic}[1]
		\Function{UnrankComposition}{$n, k, s$}
		\If{$n=k$ and $s=0$}
		\State \textbf{return} $[1, 1, \dots, 1]$
		\EndIf
		\State $s' := s$
		\If{$s' < \binom{n-2}{k-1}$}
		\State $C := $\textsc{UnrankComposition}($n-1, k, s'$)
		\State $C[len(C)] := C[len(C)]+1$
		\State \textbf{return} $C$
		\Else
		\State $s' := s' - \binom{n-2}{k-1}$
		\State $ C := $\textsc{UnrankComposition}($n-1, k-1, s'$) $\cup [1]$
        \State \textbf{return} $C$
		\EndIf
		\EndFunction
	\end{algorithmic}
\end{algorithm}

\begin{theorem}
	The function \textsc{UnrankTree} is an unranking algorithm and calling it
	with the parameters~$n$ and an uniformly sampled integer in $\{0,\dots, g_n-1\}$,
	it is an uniform sampler for size-$n$ weakly increasing Schröder trees.
\end{theorem}
\begin{proof}[Key-ideas]
	The total order for weakly increasing Schröder trees is the following.
	Let $\alpha$ and $\beta$ be two trees.
	If the size of $\alpha$ is strictly smaller than
	the one of $\beta$, we define $\alpha < \beta$.
	Let us suppose that both sizes are equal to $n$. 
	In the recursive construction, let $\tilde{\alpha}$ (and $\tilde{\gamma}_1$ be the tree
	(resp. the composition for the leaf substitution)
	building the tree $\alpha$
	(and respectively $\tilde{\beta}$ and $\tilde{\gamma}_2$ the ones associated to $\beta$).
	If the size of $\tilde{\alpha}$ is strictly greater than the one of $\tilde{\beta}$,
	we define $\alpha < \beta$.
	Let us now suppose that both sizes of $\tilde{\alpha}$ and $\tilde{\beta}$ are equal.
	By using an arbitrary order for the composition unranking,
	we can order $\alpha$ and $\beta$.
	
	This total order over the trees is satisfied by our algorithm: thus this latter is correct.
\end{proof}

\begin{theorem}
	Once the pre-computations have been done,
	the function \textsc{UnrankTree} necessitates $O(n^2)$ arithmetic operations
	to construct any tree of size $n$.
\end{theorem}
Due to the fact that usually the difference between $n$ and $k$ is very small,
a detailed analysis of the average case, or an more adapted composition unranking
should give a better complexity analysis.
In fact as we have seen before,
in a large typical tree, there are in average $n-\ln n$ internal nodes and 
thus most of them must be of arity $2$ and are given by the first term
in the latter sum defining $g_n$.

\begin{proof}[Proof-ideas]
	The main idea is the following: during a call to \textsc{UnrankTree},
	there are exactly the same number of new leaves in the tree under construction
	to the number of loops in the \texttt{while} instruction on Line 6.
	Outside this \texttt{while} block,
	the number of arithmetic operations is essentially due to the unranking algorithm for compositions.
	The actual version of this algorithm induces a quadratic complexity in the number of
	arithmetic operations.
	The unranking algorithm for the composition is based on 
	the classical result about the composition
	of $n>0$ in $k\in\{1, \dots, n\}$ parts:
	\begin{align*}
	C_{n,k} &= \binom{n-1}{k-1} \\
	&= C_{n-1,k} + C_{n-1,k-1}. 
	\end{align*}
\end{proof}

%\input{sections/concl}

%%
%% Bibliography
\bibliography{biblio}
\newpage
\appendix

\section{Appendix related strongly increasing Schröder trees: Section~\ref{sec:high_inc}}
\label{app:high_inc}

The Borel transform consists in the following transformation on ordinary generating series:
\[
\B \left(\sum_{n\geq 0} a_n z^n\right) = \sum_{n\geq 0} a_n \; \frac{z^n}{n!}.
\]
%From a functional point of view,
%the Borel transform corresponds respectively to
%\[
%	\B (f)(z) =\frac{1}{2i\pi}\int_{c-i\infty}^{c+i\infty} \dfrac{\exp(zt)}{t}f\left(\frac1t \right)dt,
%\]
%where the real constant $c$ is greater than the real part of all singularities of $f(1/t)/t$.

\begin{lemma}
	Using Borel transform formula on formal series, we easily derive the following identities:
	\begin{enumerate}[(i)]
		\item $\B (zf) (z) = \int_{0}^z \B f \mathrm{d}t$;
		\item $\B (f') (z) = (\B f)'(z) + z (\B f)''(z)$.
		%		\item (a priori inutile) $\B (\frac{f-f_0}{z}) (z) = (\B f)'(z)$
	\end{enumerate}
\end{lemma}

\noindent We are now ready to prove Proposition~\ref{prop:bivt}.
\begin{proof}[Proof of Proposition~\ref{prop:bivt}]
	Applying Borel on equation~\eqref{eq:specbiv_T2} and 
	using properties (i) and (ii) we obtain
	\[
	T(z,u)=zuT(z,u)+(1-u)\cdot \int_{0}^z T(z,u) dz -\frac{z^2}{2}+z.
	\]
	Then by differentiating by z 
	\[
	\frac{\partial (1-zu)T(z,u)}{\partial z}=\frac{\partial (1-u)\cdot \int_{0}^z T(z,u) dz -\frac{z^2}{2}+z}{\partial z}.
	\]
	Thus, after simplifications
	\[
	(1-zu)\frac{\partial T(z,u)}{\partial z}=T(z,u) -z+1 \quad \text{with } T(0,0)=1.
	\]
	Solving the differential equation gives the stated result.
\end{proof}

\vspace*{1cm}

Let us denote by $\X_n$ the random variable corresponding to the 
to number of internal nodes in increasing Schröder trees of size $n$.
Proposition~\ref{prop:var} aims at proving the mean value and the 
variance of $\X_n$.
\begin{proof}[Proof of Proposition~\ref{prop:var}]
	Recall that the mean and variance can be computed mechanically from the 
	bivariate generating function
	\begin{align*}
	\IE_{\T_n}(\X_n) &= \frac{[z^n] \partial_{u} T(z,u)_{\mid_{u=1}} }{[z^n]T(z,1)}, \text{ and}\\
	\IE_{\T_n}(\X_n^2)  &=  \frac{[z^n] \partial^2_{u} T(z,u)_{\mid_{u=1}} }{[z^n]T(z,1)} 
	+  \frac{[z^n] \partial_{u} T(z,u)_{\mid_{u=1}} }{[z^n]T(z,1)},
	\end{align*}
	where %$[z^n] \cdot$ stands for the $n$-th coefficient extraction and 
	$\cdot_{\mid_{u=1}}$ stands for the substitution of $u$ by 1.
	{\footnotesize
		\begin{align*}
		\mathbf{E}_{T_n}[\mathcal{X}] & = \frac{[z^n] \partial_{u} T(z,u) \mid_{u=1} }{[z^n]T(z,1)} \\
		& = [z^n] 
		\left( \frac{z}{(1-z)^2}-\frac{1}{1-z}\ln\left(\frac{1}{1-z}\right) \right)  \\
		&\hspace*{1cm} +  \frac{1}{2}\left(\frac{1}{1-z}-z-1\right) \\
		& = n-H_n+\frac{1}{2} 
		\end{align*}
	}
	Let $n\ge 2$, the mean value of $\X_n$ is equal to
	\[
	\IE_{\T_n}[\X_n^2] = n(n-1) -2n(H_n-1) + \sum_{k=1}^{n-1} \frac{1}{n-k} H_k,
	\]
	and thus
	\begin{align*}
	\IE_{\T_n}[\X_n^2] = &n(n-1) -2n\; \ln n  -2(\gamma -1)n   +  \ln^2 n \\
	&+2\gamma\; \ln n +\gamma^2 - \frac{\pi^2}{6} + O\left(\frac{\log n}{n}\right).
	\end{align*}
	In the same vein, when $n$ tends to infinity, we get
	\[
	\IV_{T_n}[\X] = \ln n + \gamma -\frac{\pi^2}{6}-\frac{5}{4} + O\left(\frac{\log n}{n}\right).
	\]
\end{proof}
We are now ready to prove the limit distribution for $\X_n$.
\begin{proof}[Proof of Theorem~\ref{theo:limit_law_internal}]
	This proof is an adaptation on Flajolet and Sedgewick's proof on the limit Gaussian
	law of Stirling Cycle numbers~\cite[p.~644]{FS09}
	
	We take the probability generating function of $\hat{T}_n(u)$ it is obvious that if $\frac{t_{n,k}}{t_n}$ is a limit Gaussian law then so is $\frac{\hat{t}_{n,k}}{t_n}=\frac{t_{n,n-k}}{t_n}$. We will just get the mirror of the probability the standard deviation will not change $\hat{\sigma_n}=\sigma_n$ and the mean will be the mirror mean so $\hat{\mu}_n=n-\mu_n$.
	
	\[ \hat{p}_n(u) = \frac{2u(u+2)(u+3)\dots(u+n-1)}{n!}.\]
	Thus we have 
	\[ p_n(u) = \frac{2\Gamma(u+n)}{(u+1)\Gamma(u)\Gamma(n+1)}.\]
	Near $u=1$ we find an estimate of $p_n(u)$ using Stirling formula for the Gamma function
	{\small
		\[ p_n(u) = \frac{ n^{u-1}}{\Gamma(u)}\left(1+O\left(\frac{1}{n}\right)\right)=\frac{ \left( e^{u-1} \right)^{\log n}  }{\Gamma(u)}\left(1+O\left(\frac{1}{n}\right)\right).\]
	}
	
	Now we can study the standardized random variable $\hat{X}^\star_n=\frac{\hat{\mathcal{X}}-\hat{\mu}_n  }{\hat{\sigma}_n}$.
	The standardization of a random variable can be translated directly on the characteristic function.
	\[  \phi_{X^\star_n}(t) = e^{-it\frac{\mu}{\sigma}} \phi_{X_n}(\frac{t}{\sigma}).\]
	{\footnotesize
		\begin{align*}
		\phi_{\hat{X}^\star_n}(t) = &e^{-it\frac{\log n-\gamma+\frac{1}{2}}{\sqrt{\log n + \gamma -\frac{\pi^2}{6}-\frac{5}{4}} } } \\
		& \cdot \frac{ \left( exp(\log n (e^{\frac{it}{\sqrt{\log n + \gamma -\frac{\pi^2}{6}-\frac{5}{4}}}}-1)) \right)  }{\Gamma(u)} \left(1+O\left(\frac{1}{n}\right)\right).
		\end{align*}
	}
	For a fixed $t$ and as $n \rightarrow \infty$,
	\[  \log \phi_{\hat{X}^\star_n}(t) = -\frac{t^2}{2} + O(\frac{1}{\log n}).\]
	This last result is obtained by limited development of $\log n \cos{\frac{t}{\sqrt{\log n + \gamma -\frac{\pi^2}{6}-\frac{5}{4}}}}$.
	Finally we have 
	\[    \phi_{\hat{X}^\star_n}(t) \sim e^{-\frac{t^2}{2}},\]
	which is the characteristic function of the Gaussian law.
\end{proof}

\section{Appendix related to weakly increasing Schröder trees:
	Section~\ref{sec:weak_inc}}
\label{app:weak}

Derivation for the exponential generating function
for $(g_n)$:
We have,
\[ g_n = \delta_n +  \sum\limits_{k=1}^{n-1}\binom{n-1}{k-1}g_k.\]
Where $\delta_l$ is $1$ for $l=1$ and $0$ otherwise. Adding $g_n$ to both sides gives 
\[ 2 g_n = \delta_n +  \sum\limits_{k=1}^{n}\binom{n-1}{k-1}g_k.\]
Finally multiplying both sides by $\frac{z^l}{l!}$ and summing over all $l\geq 0$
\[2G(z)=z+  \sum\limits_{l \geq 0} \frac{z^l}{l!}\sum\limits_{k=1}^{n}\binom{n-1}{k-1}g_k\]
Deriving this last equation yields to the equation of Ordered Bell number. which has been studied by different authors. 
See \cite{Pippenger10} for a derivation of the exponential generating function,
\[G(z)' = \frac{1}{2-e^z}.\]
Finally we have,
\[ G(z) = \frac{1}{2}\left( z - \ln{(2-e^z)}\right).\]

Proof of the theorem \ref{weakly:internal nodes}
We define $f_n =\sum\limits_{k=0}^{n-2}\binom{n-2}{k}(k+1)g_{k+1}$
\begin{lemma}$\frac{f_n}{g_n} \sim C$ where $C \approx 1.38$ is a constant 
\end{lemma} 
\begin{proof}
	Wilf has given an approximation of the error term of Ordered Bell numbers in~\cite{wilb} which we can use,
	\[ g_n = \frac{(n-1)!}{2\ln(2)^n}+O(\gamma^{n-1}(n-1)!),\]
	with $\gamma = \frac{1}{\sqrt{ln(2)^2 + 4\pi^2}} \approx 0.16\dots$.
	
	Remark that $\frac{1}{ln(2)} \approx 1.44\dots > \gamma$.
	Now,
	{\small
		\begin{align*}
		\frac{f_n}{g_n} & = \frac{ \sum\limits_{k=0}^{n-2}\binom{n-2}{k}(k+1) \left( \frac{k!}{2\ln(2)^{k+1}}+O(\gamma^{k}k!) \right)}{\frac{(n-1)!}{2\ln(2)^n}+O(\gamma^{n-1}(n-1)!)} \\
		%    & =  \frac{ \sum\limits_{k=0}^{n-2}\binom{n-2}{k}(k+1) \left( \frac{k!}{2\ln(2)^{k+1}} \right)}{\frac{(n-1)!}{2\ln(2)^n}+O(\gamma^{n-1}(n-1)!)} \\
		%   & \quad  + \frac{ \sum\limits_{k=0}^{n-2}\binom{n-2}{k}(k+1) \left( O(\gamma^{k}k!) \right)}{\frac{(n-1)!}{2\ln(2)^n}+O(\gamma^{n-1}(n-1)!)}\\
		& =  \frac{ \sum\limits_{k=0}^{n-2}\binom{n-2}{k}(k+1) \left( \frac{k!}{2\ln(2)^{k+1}} \right)}{\frac{(n-1)!}{2\ln(2)^n}+O(\gamma^{n-1}(n-1)!)} + O(\frac{1}{n})\\
		& =  \frac{ \sum\limits_{k=0}^{n-2}(k+1) \left( \frac{1}{\ln(2)^{k+1}} \right)}{\frac{(n-1)}{\ln(2)^n}+O(\gamma^{n-1}(n-1)!)} + O(\frac{1}{n})\\
		%       & =  \frac{ \frac{n-1}{\ln2^{n-1}} + \frac{n-2}{\ln2^{n-2}} + \frac{n-3}{2\ln2^{n-3}} + \dots + \frac{1}{(n-2)!\ln2} }{\frac{(n-1)}{\ln(2)^n}+O(\gamma^{n-1}(n-1)!)} + O(\frac{1}{n})\\
		& =  \frac{ n(\sum\limits_{k=1}^{n-1}\frac{1}{(n-1-k)!\ln2^k)}) - ( \sum\limits_{k=1}^{n-1}\frac{n-k}{(n-1-k)!\ln2^k}) }{\frac{(n-1)}{\ln(2)^n}+O(\gamma^{n-1}(n-1)!)} + O(\frac{1}{n})\\
		& =  \frac{ n(\sum\limits_{k=1}^{n-1}\frac{1}{(n-1-k)!\ln2^k)} - ( \sum\limits_{k=1}^{n-1}\frac{n-k}{(n-1-k)!\ln2^k}) }{\frac{(n-1)}{\ln(2)^n}+O(\gamma^{n-1}(n-1)!)} + O(\frac{1}{n})\\
		& =  c +  O(\frac{1}{n})
		\end{align*}
	}
\end{proof}

Then for large $n$, we can show the result by induction. Taking\\ $Gu_n=\frac{(n-1)!}{2\ln(2)^n}(n-\ln n)  +O(\gamma^{n-1}(n-1)!)$.
{\footnotesize
	\begin{align*}
	&\IE_{\G_n}(\X_n)  = \frac{Gu_n}{g_n} & \\
	&= c+ O(\frac{1}{n})+\frac{\sum\limits_{k=1}^{n-1}\binom{n-1}{k-1} \left( \frac{(k-1)!}{2\ln(2)^k}(k-\ln k  +O(\gamma^{k-1}(k-1)!) \right) }{\frac{(n-1)!}{2\ln(2)^n}+O(\gamma^{n-1}(n-1)!)} \\
	& = c+ O(\frac{1}{n})+\frac{\sum\limits_{k=1}^{n-1}\binom{n-1}{k-1} \frac{(k-1)!}{2\ln(2)^k}k - \sum\limits_{k=1}^{n-1}\binom{n-1}{k-1} \frac{(k-1)!}{2\ln(2)^k}\ln k}{\frac{(n-1)!}{2\ln(2)^n}+O(\gamma^{n-1}(n-1)!)}\\
	& \quad + \frac{\sum\limits_{k=1}^{n-1}\binom{n-1}{k-1} O(\gamma^{k-1}(k-1)!) }{\frac{(n-1)!}{2\ln(2)^n}+O(\gamma^{n-1}(n-1)!)} \\
	\end{align*}
}
Thus we get:
{\footnotesize
	\begin{align*}
	&\IE_{\G_n}(\X_n)  =\\
	& = c+ O(\frac{1}{n})+ n +c' +\frac{-\sum\limits_{k=1}^{n-1}\binom{n-1}{k-1} \frac{(k-1)!}{2\ln(2)^k}\ln n}{\frac{(n-1)!}{2\ln(2)^n}+O(\gamma^{k-1}(k-1)!)} \\
	& \quad + \frac{\sum\limits_{k=1}^{n-1}\binom{n-1}{k-1} O(\gamma^{k-1}(k-1)!) }{\frac{(n-1)!}{2\ln(2)^n}+O(\gamma^{k-1}(k-1)!)}\\
	& = c+ O(\frac{1}{n})+ n +c'- \ln n +\frac{ \sum\limits_{k=1}^{n-1}\binom{n-1}{k-1} O(\gamma^{k-1}(k-1)!) }{\frac{(n-1)!}{2\ln(2)^n}+O(\gamma^{n-1}(n-1)!)}\\
	& = c+ n +c'- \ln n +O(\frac{1}{n})\\
	& \sim n - \ln n.
	\end{align*}
}

\end{document}